\documentclass[11pt,letterpaper,twoside,reqno]{amsart}

\usepackage{subfigure}
\usepackage{amsmath}
\usepackage{amssymb}
\usepackage{bm}
\usepackage{mathrsfs}
\usepackage{color}

\usepackage{url}

\usepackage{supertabular}

\newtheorem{thm}{Theorem}
\newtheorem{cor}[thm]{Corollary}
\newtheorem{lemma}[thm]{Lemma}
\newtheorem{prop}[thm]{Proposition}

\newtheorem{defn}[thm]{Definition}

\numberwithin{equation}{section}

\newenvironment{reftheorem}[1]{\medskip\parindent 0pt{\bf Theorem \ref{#1}}\em }{\vspace{1em}}

\DeclareMathAlphabet{\mathsfsl}{OT1}{cmss}{m}{sl}

\newcommand{\term}{\emph}

	{\makebox{\phantom{}} \\ %
		\textsc{Input:}\begin{itemize}}%
	{\end{itemize}}

	{\textsc{Output:}\begin{itemize}}%
	{\end{itemize}}

	{\textsc{Procedure:}\begin{enumerate}}%
	{\end{enumerate}}

\renewcommand{\phi}{\varphi}

\newcommand{\Rspace}[1]{\mathbb{R}^{#1}}

\newcommand{\vct}[1]{{#1}}
\newcommand{\mtx}[1]{\bm{#1}}

\newcommand{\psinv}{\dagger}

\newcommand{\supp}[1]{\operatorname{supp}(#1)}

\newcommand{\norm}[1]{\left\Vert {#1} \right\Vert}

\newcommand{\pnorm}[2]{\norm{#2}_{#1}}

\newcommand{\Fee}{\mtx{\Phi}}

\newcommand{\poly}{\operatorname{poly}}
\newcommand{\polylog}{\operatorname{polylog}}

\newtheorem{claim}[thm]{Claim}

\newcommand\rtp{\otimes_{\rm r}}                 %row tensor product
\newcommand\rip[3]{{\rm RIP}_{{#1},{#2},{#3}}}   % RIP with decoration

\newcommand{\bigO}{{\rm O}}
\newcommand{\halmos}{\hspace*{\fill}\rule{1ex}{1.4ex}}
\makeatletter
\def\newproof#1{\@nprf{#1}}
\def\@nprf#1#2{\expandafter\@ifdefinable\csname #1\endcsname
\global\@namedef{#1}{\@prf{#1}{#2}}\global\@namedef{end#1}{\@endproof}}
\def\@prf#1#2{\@beginproof{#2}{\csname the#1\endcsname}\ignorespaces}
\def\@beginproof#1{\rm \trivlist \item[\hskip \labelsep{\bf #1: }]}
\def\@endproof{\halmos \endtrivlist}
\makeatother
\newcommand{\Ph}{\mtx{\Phi}}

\newcommand{\R}{\mathbb{R}}

\usepackage{rotating}
\usepackage{graphicx}

\begin{document}

\title[Unified approach to signal recovery] {Combining geometry and combinatorics: \\
a unified  approach to sparse signal recovery }

\author[Berinde, Gilbert, Indyk, Karloff, Strauss]{R.\ Berinde, A.\ C.\ Gilbert, P.\ Indyk, H.~Karloff, and M.\ J.\ Strauss}

\thanks{Berinde is with the Department of Electrical Engineering and
  Computer Science, MIT.  E-mail: \url{texel@mit.edu}.  Gilbert is
  with the Department of Mathematics, The University of Michigan at
  Ann Arbor.  E-mail: \url{annacg}\url{@umich.edu}.  Indyk is with the
  Computer Science and Artificial Intelligence Laboratory, MIT.
  E-mail: \url{indyk@theory.lcs.mit.edu}.  Karloff is with AT\&T
  Labs-Research.  E-mail: \url{howard@research.att.com}.  Strauss is
  with the Department of Mathematics and the Department of Electrical
  Engineering and Computer Science, The University of Michigan at Ann
  Arbor.  E-mail: \url{martinjs@umich.edu}.  ACG is an Alfred P. Sloan
  Research Fellow and has been supported in part by NSF DMS 0354600.
  MJS has been supported in part by NSF DMS 0354600 and NSF DMS
  0510203.  ACG and MJS have been partially supported by DARPA ONR
  N66001-06-1-2011.
  RB and PI are supported in part by David and Lucille Packard Fellowship.
  PI is partially supported by MADALGO (Center for Massive Data Algorithmics, funded by the Danish National Research Association).
}

\date{\today}

\begin{abstract}
  There are two main algorithmic approaches to sparse signal recovery:
  geometric and combinatorial.  The geometric approach starts with a
  geometric constraint on the measurement matrix $\Ph$ and then uses
  linear programming to decode information about $x$ from $\Ph x$.
  The combinatorial approach constructs $\Ph$ and a combinatorial
  decoding algorithm to match.  We present a unified approach to these
  two classes of sparse signal recovery algorithms.

  The unifying elements are the adjacency matrices of high-quality
  {\em unbalanced expanders}.  We generalize the notion of {\em
    Restricted Isometry Property} (RIP), crucial to compressed sensing
  results for signal recovery, from the Euclidean norm to the $\ell_p$
  norm for $p \approx 1$, and then show that unbalanced expanders are
  essentially equivalent to RIP-$p$ matrices.

  From known deterministic constructions for such matrices, we obtain
  new {\em deterministic} measurement matrix constructions and
  algorithms for signal recovery which, compared to previous
  deterministic algorithms, are superior in either the number of
  measurements or in noise tolerance.
\end{abstract}
\maketitle
%\thispagestyle{empty}
%\newpage

\section{Introduction}
% what are we doing: looking at sketches of data
With the rise in high-speed data transmission and the exponential
increase in data storage, it is imperative that we develop effective
data compression techniques, techniques which accomodate both the
volume and speed of data streams.  A new approach to compressing
$n$-dimensional vectors (or signals) begins with linear observations
or measurements.  For a signal $x$, its compressed representation is
equal to $\Ph x$, where $\Ph$ is a carefully chosen $m \times n$
matrix, $m \ll n$, often chosen at random from some distribution.  We
call the vector $\Ph x$ the {\em measurement vector} or a {\em sketch}
of $x$.  Although the dimension of $\Ph x$ is much smaller than that
of $x$, it retains many of the essential properties of $x$.

% why is this a good thing: linearity for data streams, physics for DSP
There are several reasons why linear compression or sketching is of
interest.  First, we can easily maintain a linear sketch $\Ph x$ under
linear updates to the signal $x$.  For example, after incrementing the
$i$-th coordinate $x_i$, we simply update the sketch as $\Ph (x+e_i) =
\Ph x + \Ph e_i$.  Similarly, we also easily obtain a sketch of a sum
of two signals given the sketches for individual signals $x$ and $y$,
since $\Ph(x+y)=\Ph x+\Ph y$.  Both properties are very useful in
several computational areas, notably computing over data streams~\cite{AMS99:Space-Frequency,Muthu:survey,SSS}, network measurement~\cite{EV},  query optimization and answering in databases~\cite{AMS99:Space-Frequency}.

Another scenario where linear compression is of key importance is {\em
  compressed
  sensing}~\cite{CRT06:Stable-Signal,Don06:Compressed-Sensing}, a
rapidly developing area in digital signal processing.  In this
setting, $x$ is a physical signal one wishes to sense (e.g., an image
obtained from a digital camera) and the linearity of the observations
stems from a physical observation process.  Rather than first
observing a signal in its entirety and then compressing it, it may be
less costly to sense the compressed version directly via a physical
process.  A camera ``senses'' the vector by computing a dot product
with a number of pre-specified measurement vectors.
See~\cite{TLWDBSKB06:CSCameraArch,DDTLTKB} for a prototype camera
built using this framework.
Other applications of linear sketching include database privacy~\cite{DMK07}.

%The sketch itself is not particularly useful; 
Although the sketch is considerably smaller than the original vector,
we can still recover a large amount of information about $x$.  
See the
surveys~\cite{Muthu:survey,SSS} on streaming and sublinear algorithms
for a broad overview of the area.  In this paper, we focus on
retrieving a {\em sparse approximation} $x_*$ of $x$.  A vector is
called $k$-{\em sparse} if it has at most $k$ non-zero elements in the
canonical basis (or, more generally, $k$ non-zero coefficients in some
basis $B$).  The goal of the sparse approximation is to find a vector
$x_*$ such that the $\ell_p$ {\em approximation error} $\|x-x_*\|_p$
is at most $C>0$ times the smallest possible $\ell_q$ approximation
error $\|x-x'\|_q$, where $x'$ ranges over all $k$-sparse vectors.
Note that the error $\|x-x'\|_q$ is minimized when $x'$ consists of
the $k$ largest (in magnitude) coefficients of $x$.

%For the algorithms given in this paper we have $p=q$. 

There are many algorithms for recovering sparse approximations (or
their variants) of signals from their sketches. The early work on this
topic includes the {\em algebraic} approach
of~\cite{Man}(cf.~\cite{GGIMS}).  Most of the known algorithms,
however, can be roughly classified as either {\em combinatorial} or
{\em geometric}.

{\bf Combinatorial approach.} In the combinatorial approach, the
measurement matrix $\Ph$ is sparse and often binary.  Typically, it
is obtained from an adjacency matrix of a sparse bipartite random
graph.  The recovery algorithm proceeds by iteratively, identifying and
eliminating ``large'' coefficients\footnote{In the non-sketching
  world, such methods algorithms are often called ``weak greedy
  algorithms'', and have been studied thoroughly by
  Temlyakov~\cite{Tem02:Nonlinear-Methods}} of the vector $x$.  The
identification uses non-adaptive binary search techniques.  Examples
of combinatorial sketching and recovery algorithms
include~\cite{GGIKMS02:Fast-Small-Space, CCF,CM03b,GKMS03:One-Pass,
  DWB05:Fast-Reconstruction,SBB06, SBB06b, CM06:Combinatorial-Algorithms,GSTV06:Chaining,GSTV07:HHS,I:extractor,XH:expander}
and others.
%\notate{Some of them only find large coefficients - should we point that out ?}
% PIOTR: only GKMS03 cannot be directly extended to find sparse approximation, 
% so the statement is morally true. Commented out.

The typical advantages of the combinatorial approach include fast
recovery, often sub-linear in the signal length $n$ if $k \ll n$ and
fast and incremental (under coordinate updates) computation of the
sketch vector $\Ph x$.  In addition, it is possible to construct
efficient (albeit suboptimal) measurement matrices {\em explicitly},
at least for simple type of signals.  For example, it is
known~\cite{I:extractor,XH:expander} how to explicitly construct
matrices with $k 2^{(\log \log n)^{O(1)}}$ measurements, for signals
$x$ that are exactly $k$-sparse.  The main disadvantage of the
approach is the suboptimal sketch length.

{\bf Geometric approach.} This approach was first proposed in the
papers~\cite{CRT06:Stable-Signal,Don06:Compressed-Sensing} and has
been extensively investigated since then (see~\cite{CS:Web} for a
bibliography).  In this setting, the matrix $\Ph$ is dense, with at
least a constant fraction of non-zero entries.  Typically, each row of
the matrix is independently selected from a sub-exponential
$n$-dimensional distribution, such as Gaussian or Bernoulli.  The key
property of the matrix $\Ph$ which yields efficient recovery
algorithms is the {\em Restricted Isometry
  Property}~\cite{CRT06:Stable-Signal}, which requires that for any
$k$-sparse vector $x$ we have $\|\Ph x\|_2 =(1 \pm \delta) \|x\|_2$.
If a matrix $\Ph$ satisfies this property, then the recovery process
can be accomplished by finding a vector $x_*$ using the following
linear program:
\[ \min\ \|x_*\|_1 \textrm{ subject to }\Ph x_* = \Ph x. \tag{P1}
\]

% of minimum $l_1$ norm, subject to $\Ph x = \Ph x_*$. This task can be accomplished using linear or convex programming.

The advantages of the geometric approach include a small number of
measurements ($O(k \log(n/2k))$ for Gaussian matrices and $O(k
\log^{O(1)}n)$ for Fourier matrices) and resiliency to measurement
errors\footnote{Historically, the geometric approach resulted also in
  the first {\em deterministic} or {\em uniform} recovery algorithms,
  where a fixed matrix $\Ph$ was guaranteed to work for {\em all}
  signals $x$.  In contrast, the early combinatorial sketching
  algorithms only guaranteed $1-1/n$ probability of correctness for
  {\em each} signal $x$.  However, the
  papers~\cite{GSTV06:Chaining,GSTV07:HHS} showed that combinatorial
  algorithms can achieve deterministic or uniform guarantees as
  well.}.  The main disadvantage is the running time of the recovery
procedure, which involves solving a linear program with $n$ variables
and $n+m$ constraints. The computation of the sketch $\Ph x$ can be
done efficiently for some matrices (e.g., Fourier); however, an
efficient sketch update is not possible. In addition, the problem of
finding an explicit construction of efficient matrices satisfying the
RIP property is open~\cite{T:blog}; the best known explicit
construction~\cite{DeV} yields $\Omega(k^2)$ measurements.

{\bf Connections.} There has been some recent progress in obtaining
the advantages of both approaches by decoupling the algorithmic and
combinatorial aspects of the problem.  Specifically, the
papers~\cite{NV07:rOMP, DM08, NT08} show that one can use {\em greedy}
methods for data compressed using {\em dense} matrices satisfying the
RIP property.  Similarly~\cite{GLR08}, using the results
of~\cite{KT07}, show that sketches from (somewhat) sparse matrices can
be recovered using linear programming.

The best results (up to $O(\cdot)$ constants) obtained prior to this
work are shown in Figure~\ref{f:table}\footnote{Some of the papers,
  notably~\cite{CM03b}, are focused on a somewhat different
  formulation of the problem.  However, it is known that the
  guarantees presented in the table hold for those algorithms as well.
  See Lecture 4 in~\cite{SSS} for a more detailed discussion.}. We
ignore some aspects of the algorithms, such as explicitness or
universality of the measurement matrices.  Furthermore, we present
only the algorithms that work for arbitrary vectors $x$, while many
other results are known for the case where the vector $x$ itself is
exactly $k$-sparse; e.g., see~\cite{TG05:Signal-Recovery,
  DWB05:Fast-Reconstruction,SBB06, Don06:Compressed-Sensing,
  XH:expander}.  The columns describe: citation; whether the recovery
algorithms hold with high probability for {\underline A}ll signals or
for {\underline E}ach signal;
% where the latter schemes work for each signal $x$ with probability $1-1/n$
sketch length; time to compute $\Ph x$ given $x$; time to update $\Ph
x$ after incrementing one of the coordinates of $x$; time\footnote{In
  the decoding time column LP=LP$(n,m,T)$ denotes the time needed to
  solve a linear program defined by an $m \times n$ matrix $\Ph$ which
  supports matrix-vector multiplication in time $T$.  Heuristic
  arguments indicate that LP$(n,m,T) \approx \sqrt{n} T$ if the
  interior-point method is employed.  In addition, the
  paper~\cite{NV07:rOMP} does not discuss the running time of the
  algorithm. Our bound is obtained by multiplying the number of
  algorithm iterations (i.e., $k$) by the number of entries in the
  matrix $\Ph$ (i.e., $nk \log^c n$).  See~\cite{NT08} for an in-depth
  discussion of the running times of OMP-based procedures.} to recover
an approximation of $x$ given $\Ph x$; approximation guarantee; and
whether the algorithm is robust to noisy measurements.
%
% PIOTR: here be the running time discussion
%
%
%
% ``OMP'' denotes the time needed to perform the Orthogonal Matching Pursuit procedure {\bf PIOTR: Need to be more specific}.  
%
%
%
%
In the approximation error column, $\ell_p \le A \ell_q$ means that
the algorithm returns a vector $x_*$ such that $\|x-x_*\|_p \le A
\min_{x'} \|x-x'\|_q$, where $x'$ ranges over all $k$-sparse vectors.
In~\cite{CDD}, the authors show that an approximation guarantee of the
form ``$\ell_2 \le \frac{C}{k^{1/2}} \ell_1$'' implies a ``$\ell_1 \le
(1+O(C)) \ell_1$'' guarantee, and that it is impossible to achieve
``$\ell_2 \le C \ell_2$'' deterministically (or for all signals
simultaneously) unless the number of measurements is $\Omega(n)$.  The
parameters $C>1$, $c \ge 2$ and $a>0$ denote absolute constants,
possibly different in each row.  We assume that $k < n/2$.

In addition, in Figure~\ref{f:table-new} we present very recent
results discovered during the course of our research.  Some of the
running times of the algorithms depend on the ``precision parameter''
$D$, which is always bounded from the above by $\|x\|_2$ if the
coordinates of $x$ are integers.

% see Section~\ref{ss:rel} for further discussion.
%
%\subsection{Related work}
%
\begin{figure}
{\footnotesize
\begin{tabular}{|c|c|c|c|c|c|c|c|}
\hline
Paper		& A/E	& Sketch length &	Encode time &	Column sparsity/ &	Decode time & 	Approx. error &		Noise \\
 & & & & Update time & & & \\
\hline
\cite{CCF,CM06:Combinatorial-Algorithms} & 	E &	$k \log^c n$ &	$n \log^c n$ &	$\log^c n$ &	$k \log^c n$ & 	$\ell_2 \le C \ell_2$ & \\
	& E	& $k \log n$ & 	$n \log n$ & 	$\log n$ & 	$n \log n$ & $\ell_2 \le  C \ell_2$ & \\
& & & & & & &  \\
\cite{CM03b} & 	E &	$k \log^c n$ &	$n \log^c n$ &	$\log^c n$ &	$k \log^c n$ & 	$\ell_1 \le C \ell_1$ & \\
	& E	& $k \log n$ & 	$n \log n$ & 	$\log n$ & 	$n \log n$ & $\ell_1 \le  C \ell_1$ & \\
& & & & & & &  \\
\cite{CRT06:Stable-Signal} & 	A & 	$k \log(n/k)$ & 	$nk \log(n/k)$ & 	$k \log (n/k)$ & LP & 
$\ell_2 \le \frac{C}{k^{1/2}}\ell_1$ & Y \\
%& & & & & & &  \\
        & A   &	$k \log^c n$ & 	$n \log n$ & $k \log^c n$ & LP & $\ell_2 \le  \frac{C}{k^{1/2}} \ell_1$ & Y \\
& & & & & &  & \\
\cite{GSTV06:Chaining} &	A & 	$k \log^c n$ & $n \log^c n$ &	$\log^c n$ & $k \log^c n$ &	
$\ell_1 \le C \log n \ell_1$ & Y \\
& & & & & &  & \\   
\cite{GSTV07:HHS} & A & 	$k \log^c n$ & $n \log^c n$ & $\log^c n$ & $k^2 \log^c n$	& $\ell_2 \le \frac{C}{k^{1/2}} \ell_1$ & \\
& & & & & & &  \\   
\cite{NV07:rOMP}
  &		A & $k  \log (n/k)$ & 	$n k \log (n/k)$ & $k \log (n/k)$	&		$n k^2 \log^c n$ &	$\ell_2 \le \frac{C (\log n)^{1/2}}{k^{1/2}} \ell_1$  & Y\\
 &		A & $k  \log^c n$ & 	$n \log n$ & $k \log^c n$	&		$n k^2 \log^c n$ &	$\ell_2 \le \frac{C (\log n)^{1/2}}{k^{1/2}} \ell_1$  & Y\\	
& & & & & &  & \\   
\cite{GLR08} & A & $k (\log n)^{c \log \log \log n}$ & $kn^{1-a}$ & $n^{1-a}$ & LP & $\ell_2 \le \frac{C}{k^{1/2}} \ell_1$ &  \\
(k ``large'') & & & & & &  & \\
& & & & & &  & \\
\hline 
  & & & & & &  & \\
This paper	&	A &	$k \log(n/k)$ & 	$n \log(n/k)$ & 	$\log(n/k)$ & LP & $\ell_1 \le C \ell_1$ &	Y\\
& & & & & &  & \\
\hline
\end{tabular}
\caption{\footnotesize Summary of the best prior results. }
\label{f:table}
}
%\vspace{-0.2cm}   FUTZ WITH AT THE END TO FIX PAGE LENGTH!
\end{figure}

\begin{figure}
{\footnotesize
\begin{tabular}{|c|c|c|c|c|c|c|c|}
\hline
Paper		& A/E	& Sketch length &	Encode time &	Update time &	Decode time & 	Approx. error &		Noise \\
\hline
\cite{DM08}	&	A &	$k \log(n/k)$ & 	$nk \log (n/k) $ & 	$k \log(n/k)$ & $nk \log(n/k)\log D $ &  $\ell_2 \le \frac{C}{k^{1/2}} \ell_1$ &	Y\\
& & & & & &  & \\
\cite{NT08}     & A &	$k \log(n/k)$ & 	$nk \log (n/k)$ & 	$k \log(n/k)$ & $nk \log(n/k)\log D$  &  $\ell_2 \le \frac{C}{k^{1/2}} \ell_1$ &	Y\\	
                & A &	$k \log^c n$ & 	$n \log n$ & 	$k \log^c n$ &  $n \log  n \log D$ &  $\ell_2 \le \frac{C}{k^{1/2}} \ell_1$ &	Y\\
& & & & & &  & \\
\cite{RI08}       &	A & $k  \log(n/k)$ & $n \log(n/k) $ & 	$\log(n/k)$  &  $n \log (n/k) $   & $\ell_1 \le C \ell_1$ & Y \\
& & & & & &  & \\
\hline
\end{tabular}
\caption{{\footnotesize Recent work.}}
\label{f:table-new}
}
%\vspace{-0.5cm}  FUTZ WITH AT THE END TO FIX PAGE LENGTH!
\end{figure}

% This is the task that we focus on in this paper.  

\subsection{Our results}
%\notate{Not stating any formal theorems. Perhaps we should state them in a row at the end of this section}.
%
%PIOTR: I put pointers to theorems in the intro
%
In this paper we give a sequence of results which indicate that the
combinatorial and geometric approaches are, in a rigorous sense,
different manifestations of a common underlying phenomenon.  This
enables us to achieve a unifying perspective on both approaches, as
well as obtaining several new concrete algorithmic results.

We consider matrices which are \emph{binary} and {\em sparse}; i.e.,
they have only a small number $d$ of ones in each column, and all the
other entries are equal to zero.  It has been shown
recently~\cite{Chan08} that such matrices cannot satisfy the RIP
property with parameters $k$ and $\delta$, unless the number of rows
is $\Omega(k^2)$.  Our first result is that, nevertheless, such
matrices satisfy a different form of the RIP property, that we call
the {\em RIP-p property}, where the $\ell_2$ norm is replaced by the
$\ell_p$ norm.  Formally, the matrix $\Ph$ satisfies
$\rip{p}{k}{\delta}$ property if for any $k$-sparse vector $x$ we have
$\|\Ph x\|_p =(1 \pm \delta) \|x\|_p$.  In particular, we show that
this property holds for $1 \le p \le 1+O(1)/\log n$ if the matrix
$\Ph$ is an adjacency matrix of a high-quality {\em unbalanced
  expander graph}.  Thus we have a large class of natural such
measurement matrices.  We also exhibit an RIP-2 matrix that is not an
RIP-1 matrix, so that, with~\cite{Chan08}, we conclude that these two
conditions are incomparable---neither one is stronger than the other.

% THEOREM: EXPANSION IMPLIES RIP1
\begin{thm}
\label{thm:ubrip}
Consider any $m \times n$ matrix $\Ph$ that is the adjacency matrix of an $(k,\epsilon)$-unbalanced expander $G=(A,B,E)$, $|A|=n$, $|B|=m$, with left degree $d$, such that $1/\epsilon,d$ are smaller than $n$.
Then the scaled matrix $\Ph / d^{1/p}$ satisfies the $\rip{p}{k}{\delta}$  property, for $1 \le p \le 1+1/\log n$ and $\delta=C \epsilon$ for some absolute constant $C>1$.
\end{thm}

The fact that the unbalanced expanders yield matrices with RIP-$p$
property is not an accident.  In particular, we show in
Section~\ref{s:ue} that any binary matrix $\Ph$ in which each column
has $d$ ones\footnote{In fact, the latter assumption can be removed
  without loss of generality.  The reason is that, from the RIP-$1$
  property alone, it follows that each column must have roughly the
  same number of ones.  The slight unbalance in the number of ones
  does not affect our results much; however, it does complicate the
  notation somewhat.  As a result, we decided to keep this assumption
  throughout the paper.}  and which satisfies RIP-$1$ property with
proper parameters, must be an adjacency matrix of a good unbalanced
expander.  That is, an RIP-$p$ matrix and the adjacency matrix of an
unbalanced expander are essentially equivalent.  Therefore, RIP-$1$
provides an interesting ``analytic'' formulation of expansion for
unbalanced graphs.  Also, without significantly improved explicit
constructions of unbalanced expanders with parameters that match the
probabilistic bounds (a longstanding open problem), we do not expect
significant improvements in the explicit constructions of RIP-$1$
matrices.
%
% THEOREM: RIP1 IMPLIES EXPANSION
%

\begin{thm}
\label{thm:rip1ub}
  Consider any $m \times n$ binary matrix $\mtx{\Ph}$ such that each
  column has exactly $d$ ones.  If for some scaling factor $S>0$ the
  matrix $S \mtx{\Ph}$ satisfies the $\rip{1}{s}{\delta}$ property,
  then the matrix $ \mtx{\Ph}$ is an adjacency matrix of an
  $(s,\epsilon)$-unbalanced expander, for
\[ \epsilon= \Bigl(1-\frac{1}{1+\delta}\Bigr)/(2-\sqrt{2}).
\]
\end{thm}

% THEOREM: RIP1 implies LP decoding

In the next step in Section~\ref{s:lp}, we show that the RIP-1
property of a binary matrix (or, equivalently, the expansion property)
alone suffices to guarantee that the linear program $P_1$ recovers a
good sparse approximation.  In particular, we show the following

\begin{thm}
\label{t:sparse}
Let $\Ph$ be an $m \times n$ matrix of an unbalanced
$(2k,\epsilon)$-expander.  Let $\alpha(\epsilon)=(2
\epsilon)/(1-2\epsilon)$.  Consider any two vectors $x, x_*$, such
that $\Ph x=\Ph x_*$, and $\|x_*\|_1 \le \|x\|_1$. Then
% If $S$ is the set of $k$ largest (in magnitude) coefficients of $x$, then
\[ \|x-x_*\|_1 \le 2/(1-2\alpha(\epsilon)) \cdot \|x-x_k \|_1.
\]
 where $x_k$ is the optimal $k$-term representation for $x$.
\end{thm}
We also provide a noise-resilient version of the theorem; see
Section~\ref{s:lp} for details.

By combining Theorem~\ref{t:sparse} with the best known probabilistic
constructions of expanders (Section~\ref{s:ue}) we obtain a scheme for
sparse approximation recovery with parameters as in
Figure~\ref{f:table}.  The scheme achieves the best known bounds for
several parameters: the scheme is deterministic (one matrix works for
all vectors $x$), the number of measurements is $O(k \log (n/k))$, the
update time is $O(\log(n/k))$ and the encoding time is $O(n \log
(n/k))$.  In particular, this provides the first known scheme which
achieves the best known measurement and encoding time bounds {\em
  simultaneously}.  In contrast, the Gaussian and Fourier matrices are
known to achieve only one optimal bound at a time.  The fast encoding
time also speeds-up the LP decoding, given that the linear program is
typically solved using the interior-point method, which repeatedly
performs matrix-vector multiplications.  In addition to theoretical
guarantees, random sparse matrices offer an attractive empirical
performance.  We show in Section~\ref{sec:experiments} that the
empirical behavior of binary sparse matrices with LP decoding is
consistent with the analytic performance of Gaussian random matrices.

In the final part of the paper, we show that adjacency matrices of
unbalanced expanders can be augmented to facilitate sub-linear time
combinatorial recovery. This fact has been implicit in the earlier
work~\cite{GSTV07:HHS,I:extractor}; here we verify that indeed the
expansion property is the sufficient condition guaranteeing
correctness of those algorithms.  As a result, we obtain an explicit
construction of matrices with $O(k2^{(\log\log n)^{O(1)}})$ rows that
are amenable to a sublinear decoding algorithm for all vectors
(similar to that in~\cite{GSTV07:HHS}).  Previous explicit
constructions for sublinear time algorithms either had $\Omega(k^2)$
rows~\cite{CM06:Combinatorial-Algorithms} or had $O(k2^{(\log\log
  n)^{O(1)}})$ rows~\cite{I:extractor,XH:expander} but were restricted
to $k$-sparse signals or their slight generalizations.  An additional (and somewhat
unexpected) consequence is that the algorithm of~\cite{I:extractor} is
simple, effectively mimicking the well-known ``parallel bit-flip''
algorithm for decoding low-density parity-check codes.
%\footnote{In the
%  following theorem and in other places in this document, $x_k$ refers
%  to the optimal $k$-term representation for $x$ rather than the
%  $k$'th coefficient of $x$. Let $|{\rm supp}(x)|$ denote the size of
%  the support of $x$, the number of non-zero coefficients.}
\begin{thm}
\label{thm:hhsp}
  Let $\epsilon>0$ be a fixed constant, and $p=1+1/\log n$. 
Consider $x \in \Rspace{n}$ and a sparsity parameter $k$.
% and a number  $\epsilon \in (0,1)$. 
 There is a measurement matrix $\mtx{\Psi}$,
  which we can construct explicitly or randomly, and an algorithm HHS$(p)$ that, given
  measurements $v =
  \mtx{\Psi}x$ of $x$, returns an approximation $\widehat x$ of $x$
  with $O(k/\epsilon)$ nonzero entries.  The approximation satisfies
\[  \|x - \widehat x\|_p \leq \epsilon k^{1/p - 1} \|x - x_k\|_1.
\]
 where $x_k$ is the optimal $k$-term representation for $x$.
\end{thm}

\iffalse
\begin{thm}
\label{thm:hhsp}
  Let $x \in \Rspace{n}$ and fix a sparsity parameter $k$ and a number
  $\epsilon \in (0,1)$. There is a measurement matrix $\mtx{\Psi}$,
  which we can construct explicitly or randomly, with the following
  property.  There is an algorithm (called HHS$(p)$) that, given
  measurements $v =
  \mtx{\Psi}x$ of $x$, returns an approximation $\widehat x$ of $x$
  with $O(k/\epsilon)$ nonzero entries.  The approximation satisfies
\[  \|x - \widehat x\|_p \leq \epsilon k^{1/p - 1} \|x - x_k\|_1.
\]
Let $R$ denote the size of the measurements for either an explicit or
random construction.  Then the HHS$(p)$ algorithm runs in time
$\poly(R)$.  For explicit constructions $R = O(k2^{(\log\log
  n)^{O(1)}})$ and for random constructions $R = O(k\polylog n)$.
\end{thm}
\fi

Figure~\ref{f:commutativediagram} summarizes the connections among all
of our results.  We show the relationship between the combinatorial and geometric approaches to sparse signal recovery 
\begin{center}
\begin{figure}[h!]
\includegraphics[width=2.75in,height=1.5in]{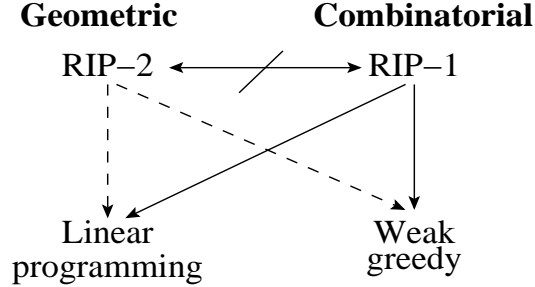}
\caption{The above diagram captures the relations between the combinatorial and
geometric approaches and the two main classes of decoding algorithms.
Connections established in prior work are shown with dashed lines. Our work
connects both approaches, with the ultimate goal of obtaining the ``best of both
worlds.''} \label{f:commutativediagram}
\vspace{-0.25cm}
\end{figure}
\end{center}
%\subsection{Relation to~\cite{RI08}}
%\label{ss:rel}
%
%Blah blah blah

%\subsection{Notation}
%For any vector $x\in\Rspace{n}$, we say that $x$ has $n$ components
% or {\em terms}.  Each term consists of an {\em index} and a {\em
%   coefficient}.  Below, we will recover indices and coefficients
% separately and we need this terminology to be clear.  For any vector
% $x$, let $x_t$ denote $x$ with all but the $t$ largest terms zeroed
% out.  Thus $x_t$ is the $t$-sparse approximation to $x$ that is
% optimal under $\ell^p$ norm ({\em i.e.}, the value of $x_\star$ that
% minimizes $\pnorm{p}{x_\star - x}$ for any $p\ge 1$.  Let $\supp x$
% denote the set of indices with non-zero coefficients in $x$ and
% $\pnorm{0}{x}=|\supp x|$ denote the number of such indices.

\section{Unbalanced expanders and RIP matrices}
\label{s:ue}

\subsection{Unbalanced expanders}

In this section we show that RIP-$p$ matrices for $p \approx 1$ can be
constructed using high-quality expanders.  The formal definition of
the latter is as follows.

\begin{defn}
  A $(k,\epsilon)$-{\em unbalanced expander} is a bipartite simple
  graph $G=(A,B,E)$ with left degree $d$ such that for any $X \subset
  A$ with $|X| \leq k$, the set of neighbors $N(X)$ of $X$ has size
  $|N(X)| \geq (1-\epsilon) d |X|$.
\end{defn}

In constructing such graphs, our goal is to make $|B|$, $d$, and
$\epsilon$ as small as possible, while making $k$ as close to $|B|$ as
possible.

The following well-known proposition can be shown using the
probabilistic method.

\begin{prop}
  For any $n/2 \ge k \ge 1$, $\epsilon>0$, there exists a
  $(k,\epsilon)$-unbalanced expander with left degree
  $d=O(\log(n/k)/\epsilon$ and right set size $O(kd/\epsilon)= O(k
  \log(n/k)/\epsilon^2$).
\end{prop}

%
% TODO
%
\begin{prop}
  For any $n \ge k \ge 1$ and $\epsilon>0$, one can explicitly
  construct a $(k,\epsilon)$-unbalanced expander with left degree
  $d=2^{ O(\log (\log(n)/\epsilon)))^3}$, left set size $n$ and right
  set size $m=k d/\epsilon^{O(1)}$.
\end{prop}

\begin{proof}
  The construction is given in~\cite{CRVW}, Theorem 7.3. Note that the
  theorem refers to notion of {\em lossless conductors}, which is
  equivalent to unbalanced expanders, modulo representing all relevant
  parameters (set sizes, degree, etc.) in the log-scale.  After an
  additional $O(nd)$-time postprocessing, we can ensure that the graph
  is simple; i.e., it contains no duplicate edges.
\end{proof}

\subsection{Construction of RIP matrices}

\begin{defn} 
  An $m\times n$ matrix $\Ph$ is said to satisfy $\rip{p}{k}{\delta}$ if,
  for any $k$-sparse vector $x$, we have
\[    \pnorm{p}{x} \le \pnorm{p}{\Ph x}\le (1+\delta)\pnorm{p}{x}.
\]
\end{defn}

Observe that the definitions of $\rip{1}{k}{\delta}$ and
$\rip{2}{k}{\delta}$ matrices are incomparable.  In what follows
below, we present sparse binary matrices with $O(k \log(n/k))$ rows
that are RIP$_{1,k,\delta}$; it has been shown recently~\cite{Chan08}
that sparse binary matrices cannot be RIP$_{2,k,\delta}$ unless the
number of rows is $\Omega(k^2)$.  In the other direction, consider an
appropriately scaled random Gaussian matrix $G$ of $R\approx k\log(n)$
rows.  Such a matrix is known to be RIP$_{2,k,\delta}$.  To see that
this matrix is {\em not} RIP$_{1,k,\delta}$, consider the signal $x$
consisting of all zeros except a single 1 and the signal $y$
consisting of all zeros except $k$ terms with coefficient $1/k$.  Then
$\|x\|_1=\|y\|_1$ but $\|Gx\|_1\approx\sqrt{k}\|Gy\|_1$.

\begin{reftheorem}{thm:ubrip}
Consider any $m \times n$ matrix $\Ph$ that is the adjacency matrix of an $(k,\epsilon)$-unbalanced expander $G=(A,B,E)$ with left degree $d$, such that $1/\epsilon,d$ are smaller than $n$.
Then the scaled matrix $\Ph / d^{1/p}$ satisfies the $\rip{p}{k}{\delta}$  property, for $1 \le p \le 1+1/\log n$ and $\delta=C \epsilon$ for some absolute constant $C>1$.
\end{reftheorem}

\begin{proof}
Let $x \in \R^n$ be a $k$-sparse vector.
Without loss of generality, we assume that the coordinates of $x$ are ordered such that $|x_1| \ge \ldots \ge |x_n|$.

The proof proceeds in two stages. 
In the first part, we show that the theorem holds for the case of $p=1$.
In the second part, we extend the theorem to the case where $p$ is slightly larger than $1$.

\ \\
{\bf The case of $p=1$.} 
We order the edges $e_t = (i_t, j_t)$, $t=1 \ldots dn$ of $G$ in a lexicographic manner. 
It is helpful to imagine that the edges $e_1, e_2 \ldots$ of $E$ are being added to the (initially empty) graph. 
An edge $e_t=(i_t,j_t)$ causes a {\em collision} if there exists an earlier edge $e_s=(i_s,j_s), s<t$, such that $j_t=j_s$.
We define $E'$ to be the set of edges which do {\em not} cause collisions, and $E''=E-E'$.

%The variable $r_t$ is equal to $-1$ if the edge $e_t$ collides with an earlier edge, and it is equal to $1$ otherwise.

\iffalse
Now a simple claim.

\begin{claim}
For any $z \in \R^l$, and $r \in \{-1,1\}^l$, with at most one entry of $r$ set to $1$, we have
\[ |r \cdot z| \ge \sum_i r_i |z_i| \]
\end{claim}

From this claim, the next one follows.
\fi

\begin{lemma}
\label{lemma:l1}
We have
\[ \sum_{(i,j) \in E''} |x_i| \le \epsilon d \pnorm{1}{x} \]
\end{lemma}

\begin{proof}
For each $t=1 \ldots dn$, we use an indicator variable $r_t \in \{0,1\}$, such that $r_t=1$ iff $e_t \in E''$.
Define a vector  $z \in \R^{dn}$ such that $z_t=|x_{i_t}|$.
Observe that 
\[  \sum_{(i,j) \in E''} |x_i| = \sum_{e_t=(i_t,j_t) \in E} r_t |x_{i_t}| = r \cdot z
\]

To upper bound the latter quantity, observe that the vectors satisfy the following constraints:
\begin{itemize}
\item The vector $z$ is non-negative.
\item The coordinates of $z$ are monotonically non-increasing.
\item For each {\em prefix set} $P_i=\{1 \ldots di\}$, $i \le k$, we have $\|r_{|P_i}\|_1 \le \epsilon di $  - this follows from the expansion properties of the graph $G$.
\item $r_{|P_1}=0$, since the graph is simple.
\end{itemize}

It is now immediate that for any $r,z$ satisfying the above constraints, we have $r \cdot z \le  \|z\|_1 \epsilon$.
Since $\|z\|_1=d \|x\|_1$, the lemma follows.
\end{proof}

Lemma~\ref{lemma:l1} immediately implies that $\pnorm{1}{ \Ph x} \ge d \pnorm{1}{x}(1-2\epsilon)$.
Since for any $x$ we have $\pnorm{1}{ \Ph x} \le d \pnorm{1}{x}$, it follows that $\Ph/d$ satisfies the $\rip{1}{k}{2 \epsilon}$ property.

\ \\
{\bf The case of $p \le 1+1/\log n$.}
Let $u = \Ph x$.
We will show that if $\epsilon$ is small enough, then the value of $\pnorm{p}{u}^p$ is close to $d \pnorm{p}{x}^p$.

We start from a few useful technical claims.

\begin{claim}
\label{claim:1}
For any $a,b \in \R$, we have $|(a+b)|^p \ge |a|^p - p |a|^{p-1}|b|$.
\end{claim}

\begin{claim}
\label{claim:2}
For any $a \in \R$ and $|b| \le |a|$, we have $|(a+b)|^p \le |a|^p + p |2a|^{p-1}|b|$.
\end{claim}

\iffalse
\begin{claim}
\label{claim:3} Let $m=n^{O(1)}$. Then, for any $z \in \R^m$ we have
\[ \pnorm{1}{z}^p \le O(1)  \pnorm{p}{z}^p
\end{claim}

For each $j$, we define a (possibly empty) sequence $S_j = i_j(1), i_j(2), \ldots$ such that (a) for all $i \in S_j$, there exists an edge $(i,j)$ in the graph $G$, and (b) we have $|x_{i_j(1)}| \ge |x_{i_j(2)}| \ge \ldots$.
\fi

For any $j \in B$, define $u'_j=x_i$  if $(i,j) \in E'$, and $u'_j=0$ otherwise. 
Also, define $u_j''=u_j-u'_j$.

\begin{claim}
\label{claim:up}
We have  $\sum_j |u_j''|^p = \Theta(\epsilon)d \pnorm{p}{x}^p$.
\end{claim}

\begin{proof}
By Lemma~\ref{lemma:l1} we know that $\sum_j |u'_j| = \sum_{ (i,j) \in E'} |x_i| \ge (1-\epsilon)d \pnorm{1}{x}$.
Therefore
\[  \sum_j |u_j''|^p 
\le  (\sum_j |u_j''|)^p 
\le  (\epsilon d \pnorm{1}{x})^p
\le \epsilon d^p n^{(1-1/p)p} \pnorm{p}{x}^p
= \Theta(\epsilon)d \pnorm{p}{x}^p 
\]
\end{proof}

% and $\sum_j |u''_j| \le \epsilon)d \pnorm{1}{x}$.

\begin{lemma}
\label{lemma:lp-lower}
We have $\pnorm{p}{u}^p \ge (1-\Theta(\epsilon)) d \pnorm{p}{x}^p$.
\end{lemma}

\begin{proof}

By Claim~\ref{claim:1} we know that 
\[  \pnorm{p}{u}^p = \sum_j |u'_j+u_j''|^p 
\ge \sum_j |u'_j|^p - \sum_j p |u'_j|^{p-1} |u_j''| 
\]

Define the set $S=\{j: \epsilon |u_j'| \le |u_j''| \}$.
We have
\begin{eqnarray*}
\sum_j p |u'_j|^{p-1} |u_j''| & = & \sum_{j \in S} p |u'_j|^{p-1} |u_j''| + \sum_{j \notin S} p |u'_j|^{p-1} |u_j''| 
  \le  (1/\epsilon)^{p-1} p \sum_{j \in S} |u_j''|^p 
      +  p \sum_{j \notin S} \epsilon |u'_j|^p \\
&  \le & \Theta\left[ \sum_j |u_j''|^p +   \epsilon \sum_j |u'_j|^p \right] 
  \le   \Theta(\epsilon)d \pnorm{p}{x}^p   +  \epsilon   \pnorm{p}{u'}^p 
\end{eqnarray*}
Therefore 
\[  \pnorm{p}{u}^p \ge \pnorm{p}{u'}^p -  \Theta(1) d \epsilon \pnorm{p}{x}^p -  \Theta( \epsilon)  \pnorm{p}{u'}^p =  \pnorm{p}{u'}^p (1-   \Theta(\epsilon) ) - \Theta(1) d \epsilon \pnorm{p}{x}^p. \]
It suffices to bound $\pnorm{p}{u'}^p$ from below.
\begin{eqnarray*}
\pnorm{p}{u'}^p & \ge & d \pnorm{p}{x}^p - \sum_{(i,j) \in E''} |x_i|^p 
  \ge   d \pnorm{p}{x}^p - \left(\sum_{(i,j) \in E''} |x_i|\right)^p \\
& \ge&  d \pnorm{p}{x}^p - \Theta(\epsilon) d \pnorm{1}{x}^p 
  \ge   d \pnorm{p}{x}^p - \Theta(\epsilon) d \pnorm{p}{x}^p \\
& = & (1- \Theta(\epsilon)) d \pnorm{p}{x}^p
\end{eqnarray*}
where we used Lemma~\ref{lemma:l1} in the third line. 
Altogether
\[  \pnorm{p}{u}^p \ge  (1- \Theta(\epsilon) ) (1- \Theta(\epsilon)) d \pnorm{p}{x}^p- \Theta(1) d \epsilon \pnorm{p}{x}^p =  (1- \Theta(\epsilon)) d \pnorm{p}{x}^p
\]
\end{proof}

\begin{lemma}
\label{lemma:lp-upper}
We have $\pnorm{p}{u}^p \le (1+\Theta(\epsilon)) d \pnorm{p}{x}^p$.
\end{lemma}

\begin{proof}

Define the set $T=\{j: |u_j''| \le |u_j'| \}$.
Decompose $\pnorm{p}{u}^p$ into a sum
\[  \sum_{j \in T} |u_j|^p +\sum_{j \notin T} |u_j|^p = T_1+T_2 \]
   
By Claim~\ref{claim:2} we know that 
\[  T_1= \sum_{j \in T} |u_j|^p = \sum_{j \in T} |u'_j+u_j''|^p 
\le \sum_{j \in T} |u'_j|^p + \sum_{j \in T} p |2u'_j|^{p-1} |u_j''| 
\]
which, as in the proof of Lemma~\ref{lemma:lp-lower}, can be bounded from  above by 
\[  \sum_{j \in T} |u'_j|^p (1+\Theta(\epsilon)) + \Theta(\epsilon) d \pnorm{p}{x}^p \]

At the same time, using Claim~\ref{claim:up}
\[T_2 = \sum_{j \notin T} |u_j|^p \le   \sum_{j \notin T} |2u_j''|^p =\Theta(\epsilon)d \pnorm{p}{x}^p \]

The lemma follows.
\end{proof}

Combining Lemma~\ref{lemma:lp-upper} and Lemma~\ref{lemma:lp-lower} yields the proof of the theorem.
\end{proof}

The above theorem shows that one can construct RIP-$p$ matrices for $p
\approx 1$ from good unbalanced expanders.  In following, we show that
this is not an accident: any binary matrix $\Phi$ in which satisfies
RIP-$1$ property with proper parameters, and with each column having
exactly $d$ ones, must be an adjacency matrix of a good unbalanced
expander.  This reason behind this is that if some set of vertices
does {\em not} expand too well, then there are many collisions between
the edges going out of that set.  If the signs of the coordinates
``following'' those edges are different, the coordinates will cancel
each other out, and thus the $\ell_1$ norm of a vector will not be
preserved.

The assumption that each column has exactly $d$ ones is not crucial,
since the RIP-$1$ property itself implies that the number of ones in
each column can differ by at most factor of $1+\delta$.  All proofs in
this paper are resilient to this slight unbalance.  However, we
decided to keep this assumption for the ease of notation.

\begin{reftheorem}{thm:rip1ub}
Consider any $m \times n$ binary matrix $\mtx{\Phi}$ such that each
column has exactly $d$ ones.  If for some scaling factor $S>0$ the
matrix $S \mtx{\Phi}$ satisfies the $\rip{1}{s}{\delta}$ property,
then the matrix $ \mtx{\Phi}$ is an adjacency matrix of an
$(s,\epsilon)$-unbalanced expander, for
\[ \epsilon= \Big(1-\frac{1}{1+\delta}\Big)/\Big(2-\sqrt{2}\Big). \]
\end{reftheorem}
\vspace{-0.25cm}
Note that, for small values of $\delta>0$, we have
$\Big(1-\frac{1}{1+\delta}\Big)/(2-\sqrt{2}) \approx \delta/(2-\sqrt{2})$.

\begin{proof}
  Let $G=(A,B,E)$ be the graph with adjacency matrix $\mtx{\Phi}$.
  Assume that there exists $X \subset A$, $|X| =k' \le k$ such that
  $|N(X)|<dk'(1-\epsilon)$.  We will construct two $n$-dimensional
  vectors $y,z$ such that $\pnorm{1}{y} = \pnorm{1}{z}=k'$, but
  $\pnorm{1}{ \Ph z}/\pnorm{1}{\Ph y} \le 1-\epsilon (2-\sqrt{2})$,
  which is a contradiction.

  The vector $y$ is simply the characteristic vector of the set $X$.
  Clearly, we have $\pnorm{1}{y}=k'$ and $\pnorm{1}{\Ph y}=dk$.

  The vector $z$ is defined via a random process.  For $i \in X$,
  define $r_i$ to be i.i.d. random variables uniformly distributed
  over $\{-1,1\}$.  We define $z_i=r_i$ if $i \in X$, and $z_i=0$
  otherwise.  Note that $\pnorm{1}{z}=\pnorm{1}{y} =k'$.

  Let $C \subset N(X)$ be the ``collision set'', i.e., the set of all
  $j \in N(X)$ such that the number $u_j$ of the edges from $j$ to $X$
  is at least $2$.  Let $|C|=l$.  By the definition of the set $C$ we
  have $\sum_j u_j \ge 2l$.  Moreover, from the assumption it follows
  that $\sum_j u_j \ge 2 \epsilon dk'$.

  Let $v=\Ph z$.  We split $v$ into $v_C$ and $v_{C^c}$.  Clearly,
  $\pnorm{1}{v_{C^c}}=k'd- \sum_j u_j$.  It suffices to show that
  $\pnorm{1}{v_C}$ is significantly smaller than $\sum_j u_j$ for {\em
    some} $z$.

\begin{claim}
\label{claim:exp}
The expected value of $\pnorm{2}{v_C}^2$ is equal to $\sum_j u_j$.
\end{claim}
\begin{proof}
  For each $j \in C$, the coordinate $v_j$ is a sum of $u_j$
  independent random variables uniformly distributed over $\{-1,1\}$.
  The claim follows by elementary analysis.
\end{proof}

By Claim~\ref{claim:exp} we know that there {\em exists} $z$ such that
$\pnorm{2}{v_C} \le \sqrt{\sum_j u_j} \le \frac{\sum_j
  u_j}{\sqrt{2l}}$. This implies that $\pnorm{1}{v_C} \le \sqrt{l}
\pnorm{2}{v_C} \le \frac{\sum_j u_j}{\sqrt{2}}$.  Therefore
\begin{eqnarray*}  
\pnorm{1}{v} & \le &  \pnorm{1}{v_C} + \pnorm{1}{v_{C^c}}
 \le   \frac{\sum_j u_j}{\sqrt{2}} + dk'- \sum_j u_j
 =  dk' - (1-1/\sqrt{2})\sum_j u_j\\
& \le & dk' -  (1-1/\sqrt{2})\cdot 2 \epsilon dk'
 =  dk'[1-\epsilon (2-\sqrt{2})]
\end{eqnarray*}
\end{proof}

\section{LP decoding}
\label{s:lp}

In this section we show that if $A$ is an adjacency matrix of an expander graph, then the LP decoding procedure can be used for recovering sparse approximations.
%The rigorous guarantee is given by the following theorem.

Let $\Phi$ be an $m \times n$ adjacency matrix of an unbalanced $(2k,\epsilon)$-expander $G$ with left degree $d$.
%Let $\Gamma=\{y: \Phiy=0\}$.
Let $\alpha(\epsilon)=(2 \epsilon)/(1-2\epsilon)$.
We also define $E(X:Y)= E \cap (X \times Y)$ to be the set of edges between the sets $X$ and $Y$.

\subsection{L1 Uncertainty Principle}

In this section we show that any vector from the kernel of a an
adjacency matrix $\Ph$ of an expander graph is ``smooth'', i.e., the
$\ell_1$ norm of the vector cannot be concentrated on a small subset
of its coordinates.  An analogous result for RIP-2 matrices and with
respect to the $\ell_2$ norm has been used before
(e.g.,~in~\cite{KT07}) to show guarantees for LP-based recovery
procedures.

\begin{lemma}
\label{l:spread}
Consider any $y \in \R^n$ such that $\Ph y=0$, and let $S$ be any set
of $k$ coordinates of $y$.  Then we have
\[ \|y_S\|_1 \le \alpha(\epsilon) \|y\|_1. \]
\end{lemma}

\begin{proof}
  Without loss of generality, we can assume that $S$ consists of the
  largest (in magnitude) coefficients of $y$.  We partition
  coordinates into sets $S_0, S_1, S_2, \ldots S_t$, such that (i) the
  coordinates in the set $S_l$ are not larger (in magnitude) than the
  coordinates in the set $S_{l-1}$, $l \ge 1$, and (ii) all sets but
  $S_t$ have size $k$.  Therefore, $S_0=S$.  Let $\Ph '$ be a
  submatrix of $\Ph $ containing rows from $N(S)$.

  From the equivalence of expansion and RIP-1 property we know that
  $\|\Ph 'y_S\|_1 = \|\Ph y_S\|_1 \ge d (1-2\epsilon) \|y_S\|_1$.
  At the same time, we know that $\|\Ph ' y\|_1 =0$.  Therefore
\begin{eqnarray*}
  0 =   \|\Ph ' y\|_1 &  \ge &  \|\Ph 'y_S\|_1 - \sum_{l\ge1} \sum_{ (i,j) \in E, i \in S_l, j \in N(S)} |y_i| \\
  & \ge &    d (1-2\epsilon) \|y_S\|_1- \sum_{l \ge 1}  |E(S_l:N(S))| \min_{i \in S_{l-1}} |y_i| \\
  & \ge &   d (1-2\epsilon) \|y_S\|_1- \frac1k\sum_{l \ge 1}  |E(S_l:N(S))| \cdot \|y_{S_{l-1}}\|_1
\end{eqnarray*}

From the expansion properties of $G$ it follows that, for $ l \ge 1$,
we have $|N(S \cup S_l)| \ge d(1-\epsilon)|S \cup S_l|$.  It follows
that at most $d \epsilon 2k$ edges can cross from $S_l$ to $N(S)$, and
therefore
\begin{eqnarray*}
0 & \ge & d (1-2\epsilon) \|y_S\|_1- \frac1k\sum_{l \ge 1}  |E(S_l:N(S))| \cdot \|y_{S_{l-1}}\|_1 \\
& \ge &   d (1-2\epsilon) \|y_S\|_1-  d \epsilon 2 \sum_{l \ge 1}  \|y_{S_{l-1}}\|_1/k\\
& \ge &   d (1-2\epsilon) \|y_S\|_1 - 2 d \epsilon \|y\|_1
\end{eqnarray*}

It follows that $d (1-2\epsilon) \|y_S\|_1 \le 2 d \epsilon \|y\|_1$,
and thus $\|y_S\|_1 \le (2 \epsilon)/(1-2\epsilon) \|y\|_1$.
\end{proof}

\subsection{LP recovery}
\label{ss:lprec}

The following theorem provides recovery guarantees for the program $P_1$, by setting $u=x$ and $v=x_*$.

\begin{reftheorem}{t:sparse}
%\label{t:sparse}
  Consider any two vectors $u, v$, such that for $y=v-u$ we have $\Ph
  y=0$, and $\|v\|_1 \le \|u\|_1$.  Let $S$ be the set of $k$ largest
  (in magnitude) coefficients of $u$, then
\[ \|v-u\|_1 \le 2/(1-2\alpha(\epsilon)) \cdot \|u - u_{S} \|_1 
\]
\end{reftheorem}

\begin{proof}
We have

\begin{eqnarray*}
\|u\|_1 \ge \|v\|_1 &  = & \|(u+y)_S\|_1 + \|(u+y)_{S^c}\|_1 \\
& \ge & \|u_S\|_1 - \|y_S\|_1 + \|y_{S^c} \|_1 - \|u_{S^c}\|_1 \\
& = & \|u\|_1 - 2 \|u_{S^c}\|_1 + \|y\|_1 - 2 \|y_S\|_1  \\
& \ge &  \|u\|_1 - 2 \|u_{S^c}\|_1 + (1 -2 \alpha(\epsilon)) \|y\|_1 
\end{eqnarray*}
where we used Lemma~\ref{l:spread} in the last line.
It follows that
\[ 2 \|u_{S^c}\|_1 \ge  (1 -2 \alpha(\epsilon)) \|y\|_1 \]
%Therefore,
%\[
%\|v -u_S \|_1 = \|y + u_{S^c} \|_1 \le \|y\|_1 +  \|u_{S^c} \|_1 \le [1+2/(1-2\alpha(\epsilon))] \cdot \|u_{S^c} \|_1 
%\]
\end{proof}

\begin{thm}
%\label{t:sparse2}
Consider any two vectors $u, v$, such that for $y=v-u$ we have $\|\Ph y\|_1 = \beta\ge 0$, and $\|v\|_1 \le \|u\|_1$.
Let $S$ be the set of $k$ largest (in magnitude) coefficients of $u$.
Then
\[ \|v-u_S\|_1 \le 2/(1-2\alpha(\epsilon)) \cdot \|u_{S^c} \|_1 + \frac{2\beta}{d (1-2\epsilon)(1-2\alpha)}
\]

\end{thm}

\begin{proof}
Analogous to the proof of Theorem~\ref{t:sparse}. 
\end{proof}

\section{Experimental Results}
\label{sec:experiments}
\begin{figure}[h]
\centerline{
\mbox{\includegraphics[width=3.3in]{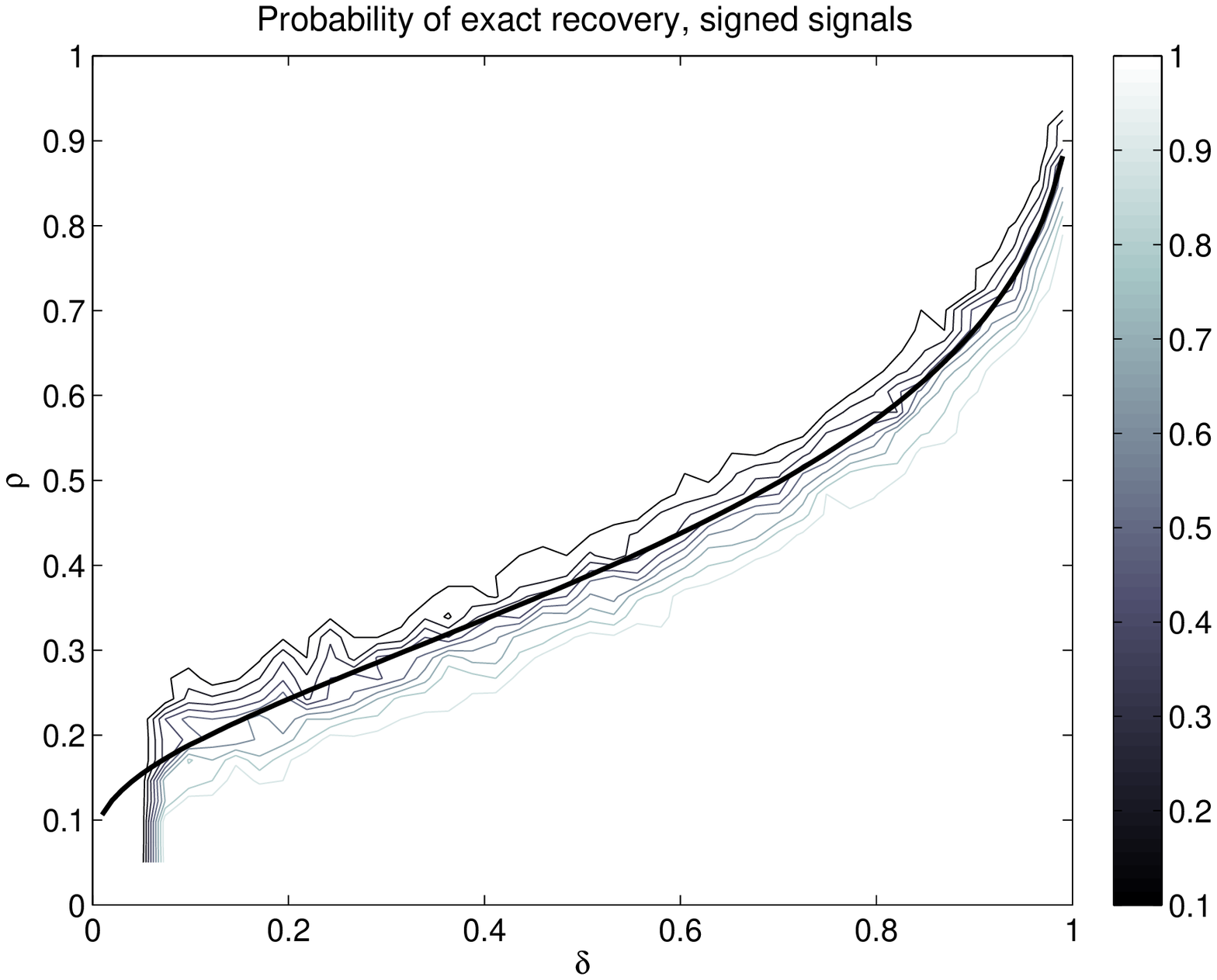}}
\mbox{\includegraphics[width=3.3in]{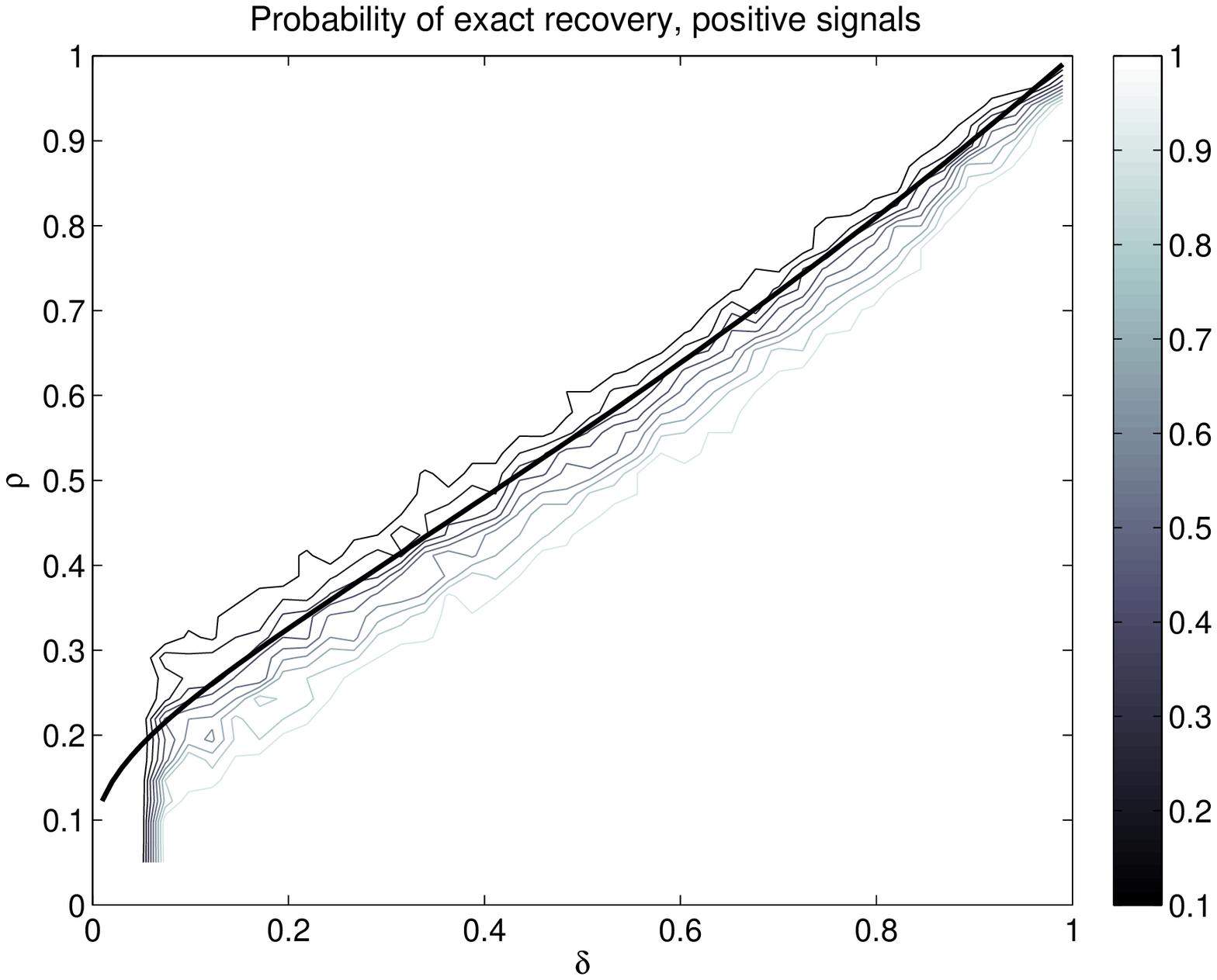}}
}
\caption{\label{f:sparse}\small Probability of correct signal recovery of a
  random $k$-sparse signal \mbox{$x \in \{-1,0,1\}^n$} (left) and $x \in \{0,1\}^n$
  (right) as a function of $k=\rho m$ and $m=\delta n$, for $n = 200$. The
  probabilities were estimated by partitioning the domain into $40 \times 40$
  data points and performing 50 independent trials for each data point, using
  random sparse matrices with $d = 8$. The thick curve demarcates a phase
  transition in the ability of LP decoding to find the sparsest solution to $G
  x_* = G x$ for $G$ a Gaussian random matrix (see \cite{DT06:Thresholds}).
  The empirical behavior for binary sparse matrices is consistent with the
  analytic behavior for Gaussian random matrices.}

\end{figure}
Our theoretical analysis shows that, up to constant factors, our
scheme achieves the best known bounds for sparse approximate recovery.
In order to determine the exact values of those constant factors, we
show, in Figure~\ref{f:sparse}, the empirical probability of correct
recovery of a random $k$-sparse signal of dimension $n=200$ as a
function of $k=\rho m$ and $m=\delta n$.  As one can verify, the
empirical $O(\cdot)$ constants involved are quite low.  The thick
curve shows the analytic computation of the phase transition between
the survival of typical $l$-faces of the cross-polytope (left) and the
polytope (right) under projection by $G$ a Gaussian random matrix.
This line is equivalent to a phase transition in the ability of LP
decoding to find the sparsest solution to $G x_* = G x$, and, in
effect, is representative of the performance of Gaussian matrices in
this framework (see \cite{Don:Poly} and \cite{DT06:Thresholds} for
more details).  Gaussian measurement matrices with $m = \delta n$ rows
and $n$ columns can recover signals with sparsity $k = \rho m$ below
the thick curve and cannot recover signals with sparsity $k$ above the
curve.  This figure thus shows that the empirical behavior of binary
sparse matrices with LP decoding is consistent with the analytic
performance of Gaussian random matrices.  Furthermore, the empirical
values of the asymptotic constants seem to agree.  See~\cite{BI08} for
further experimental data.

\newcommand{\Ps}{\mtx{\Psi}}
\section{Sublinear-time decoding of sparse vectors}
\label{sec:sparse-decoding}

In this section we focus on the recovery problem for the case where
the signal vector $x$ is $k$-sparse.  The algorithm given here is
essentially subsumed by the HHS algorithm provided in the Appendix
(modulo the slightly higher reconstruction time and the number of
measurements of the latter).  The advantage of the algorithm from this
section, however, is that it is very simple to present as well as
analyze.  This will enable us to illustrate the basic concepts without
getting into technical details needed for the case of general signals
$x$.

The algorithm we present here is essentially a simplification of the
algorithm given in~\cite{I:extractor}.  The main difference is the
assumption about the measurement matrix $\Ph$.  In~\cite{I:extractor}
the matrix $\Ph$ was assumed to be an ``augmented'' version of the
adjacency matrix of an {\em extractor.}  Here, we assume that is an
``augmented'' version of the adjacency matrix of an unbalanced
expander (or alternatively, by Theorem~\ref{thm:rip1ub}, any binary
matrix that satisfies an RIP-1 property).  The latter assumption turns
out to be the ``right'' one, resulting in further simplification of
the algorithm.

To define the matrices, we start with a straightforward definition.
If $q$ and $r$ are 0--1 vectors, we can view them as masks that
determine which entries of a signal appear and which are zeroed out.
For example, the signal $q \circ x$ consists of the entries of $x$
restricted to the nonzero components of $q$.  The notation $\circ$
indicates the componentwise or Hadamard product of two vectors.  Given
0--1 matrices $\mtx{Q}$ and $\mtx{R}$, we can form a matrix that
encodes sequential restrictions by all pairs of their rows.  We
express this matrix as the {\em row tensor product} $\mtx{Q} \rtp
\mtx{R}$.

The construction of the measurement matrix is as follows.  Let
$G=(A,B,E)$ be a $(k,\epsilon)$-unbalanced expander, $\epsilon=1/8$,
with left degree $d$, $|A|=n$, $|B|=m$.  Without loss of generality we
can assume that $m$ is a power of $2$.  Let $\Ps$ be the adjacency
matrix of $G$.  The measurement matrix $\Ph = \Ps \rtp \mtx{B}$ has $m
\log n$ rows.  The matrix $\mtx{B}$ is the {\em bit-test matrix}; its
$t$th column contains the binary expansion of $t$. That is, For $t=0
\ldots \log n-1$ and $j=0 \ldots m-1$, let $l=t+j \log n$.  The
$(l+1)$-th row $\Ph_{l+1}$ of $\Ph$ is such that, for any $i=0 \ldots
m-1$, $(\Ph_{l+1})_{i+1}=(\Ps)_{j+1} \cdot \mbox{bin}_t(i)$, where
$\mbox{bin}_t(i)$ is the $t$-th least significant bit in the binary
representation of $i$.

The purpose of the above augmentation is as follows.  Consider a
specific $k$-sparse vector $x$.  Let $\supp x$ denote the
 the support of $x$, i.e.,  the set of indices of the non-zero coordinates of $x$.
 Assume that the $(j+1)$-th row $r$ of
the matrix $\Ps$ is such that $|\supp x \cap \supp r |=1$, and let
$i+1 \in \supp x \cap \supp r$.  Then, from the values of $\Ph_{j \log
  n +1} x, \ldots, \Ph_{j \log n +\log n} x$, we can recover both $i$
and $x_{i+1}$.  This is because for each of those rows, the value of
$\Ph_{j \log n +t+1}x$ is equal to $0$ if $\mbox{bin}_t(i)=0$, or
$x_{i+1}$ otherwise; moreover, at least one of the values is non-zero.
Of course, if $|\supp x \cap \supp r| \neq 1$, then the algorithm
might ``recover'' an incorrect pair $(i,val)$.  This can be detected
if $val=0$; otherwise, the error might be undetected.

\begin{thm}
  Let $\Ph$ be a matrix described above.  Then for any $k$-sparse $x$,
  given $\Ph x$, we can recover $x$ in time $O(m \log^2 n)$.
\end{thm}

\begin{proof}
  The main component of the recovery algorithm is the procedure ${\sc
    Reduce}( \Ph x)$, which returns a vector $y$ such that $\|x-y\|_0
  \le \|x\|_0/2$.  The procedure maintains a vector $votes[\cdot]$ of
  multisets.  Initially, all entries are set to $\emptyset$.

\begin{center}
\fbox{ \tt
\begin{minipage}{6 in} 
Reduce($ \Ph x$):\\

\ \ for $j=1 \ldots m$

 \ \ \ \  compute $(l,val)$ such that if $\supp \Ph_j \cap \supp x=\{i\}$ then $x_i=val$ 

 \ \ \ \ if $val \neq 0$  then $votes[l] = votes[l] \cup \{val\}$

\ \ $y=0$

\ \ for $i=1 \ldots n$ such that $votes[i] \neq \emptyset$

\ \ \ \ if $votes[i]$ contains $ \ge d/2$ copies of  $val$
then $y_i=val$

return $y$
\end{minipage}
}
\end{center}

\begin{claim}
For any $k$-sparse vector $x$, the procedure {\sc Reduce} returns a vector $y$ such that $x-y$ is $k/2$-sparse.
\end{claim}

\begin{proof}
  From the expanding properties of $G$, it follows that at most
  $\epsilon d \|x\|_0$ rows of $\Ps$ return an incorrect pair
  $(i,val)$.  Each set of $d/2$ incorrect pairs can change the outcome
  for at most $2$ positions of $y$.  Thus, at most $4\epsilon \|x\|_0
  =\|x\|_0/2$ positions of $y$ can be incorrect.
\end{proof}

The final algorithm invokes {\sc Recover}$(\Ph x)$.

\begin{center}
\fbox{ \tt
\begin{minipage}{3 in}
Recover($\Ph x$)

\ \ if $\Ph x=0$ then return $0$

\ \ $y =$ Reduce($\Ph x$)

\ \ $\{$  We now need to recover $x-y \}$

\ \ $z=$Recover$(\Ph(x-y))$

\ \ return $y+z$
\end{minipage}
}
\end{center}

\end{proof}
\section{Conclusion}
\label{sec:conclusion}

We show in this paper that the geometric and the combinatorial
approaches to sparse signal recovery are different manifestations of a
common underyling phenomenon.  Thus, we are able to show a unified
perspective on both approaches---the key unifiying elements are the
adjacency matrices of unbalanced expanders.

In most of the recent applications of {\em compressed sensing}, a
physical device instantiates the measurement of $x$ and, as such,
these applications need measurement matrices which are conducive to
physical measurement processes.  This paper shows that there is
another, quite different, large, natural class of measurement
matrices, combined with the same (or similar) recovery algorithms for
sparse signal approximation.  These measurement matrices may or may
not be conducive to physical measurement processes but they are quite
amenable to computational or digital signal measurement.  Our work
suggests a number of applications in digital or computational
``sensing'' such as efficient numerical linear algebra and network
coding.

The preliminary experimental analysis exhibits interesting
high-dimensional geometric phenomena as well.  Our results suggest
that the projection of polytopes under Gaussian random matrices is
similar to that of projection by sparse random matrices, despite the
fact that Gaussian random matrices are quite different from sparse
ones.

% Conclusion and
%some open problems: (i) mathematical problems such as preservation of
%faces of polytopes under sparse matrices as opposed to orthoprojectors
%or Gaussian random matrices, analytic definition of expansion (does
%this get you anywhere interesting?), necessary and sufficient
%conditions for robust LP decoding, no-go theorems for RIP-1 => l2/l1
%error guarantees (perhaps covered already by DeVore's results),
%alternative constructions for RIP-1 matrices, show that RIP-1 means
%the matrix is essentially sparse and binary; (ii) algorithmic
%``applications'' such as network coding, numerical computing with l1
%guarantees, funny matrix multiplication, subroutine for determining
%which dimensions/directions/derivatives/variables/rows/columns are
%important for some other larger computation; (iii) scientific
%applications such as group testing in biological data.

{\bf Acknowledgments:} The authors would like to thank: Venkat
Guruswami and Salil Vadhan, for their help on the expanders front; David Donoho and Jared Tanner for providing the data for the analytic
Gaussian treshold curve in Figure \ref{f:sparse}; Justin Romberg for his help and clarifications regarding the $\ell_1$-MAGIC package;  and Tasos Sidiropoulos for many helpful comments.
% The authors would like to thank Venkat
%Guruswami and Salil Vadhan for their help on the expanders front. They also
%thank David Donoho and Jared Tanner for providing the data for the analytic
%Gaussian threshold curve in Figure~\ref{f:sparse}.

\bibliographystyle{alpha}
\bibliography{rip2expandbib}

\appendix
\section{HHS$(p)$ decoding}

In this section we focus on the general recovery problem for an
arbitrary vector $x$.  We show that, as in the previous algorithm in
Section~\ref{sec:sparse-decoding}, we can use the adjacency matrices
of unbalanced expanders, suitably augmented and concatenated, to
obtain an explicit construction of matrices that are designed for
the sub-linear time combinatorial recovery algorithm HHS
in~\cite{GSTV07:HHS}.  We call this modified algorithm the HHS$(p)$
algorithm.  For the sake of exposition, we highlight the necessary
changes in the construction of the matrices and the algorithm and
summarize the remaining, unchanged portions of the algorithm.
Similarly, for the analysis of HHS$(p)$, we present only those
portions which are significantly different from the original analysis
and summarize those which are not.  We establish the following result.

\begin{reftheorem}{thm:hhsp}
Let $x \in \Rspace{n}$ and fix a sparsity parameter $k$ and a number $\epsilon \in (0,1)$. There is a measurement matrix $\mtx{\Psi}$, which we can construct explicitly or randomly, with the following property.  The HHS$(p)$ algorithm, given measurements $v = \mtx{\Psi}x$ of $x$, returns an approximation $\widehat x$ of $x$ with $O(k/\epsilon)$ nonzero entries.  The approximation satisfies
\[  \|x - \widehat x\|_p \leq \epsilon k^{1/p - 1} \|x - x_k\|_1.
\]
Let $R$ denote the size of the measurements for either an explicit or random construction.  Then, the HHS$(p)$ algorithm runs in time $\poly(R)$.  
\end{reftheorem}

For explicit constructions $R = O(k2^{\log\log n^{O(1)}})$ and for random constructions $R = O(k\polylog n)$.  

\begin{cor}
If we truncate the output $\widehat x$ of the algorithm to the $k$ largest terms, then the modified output $\widehat x_k$ satisfies
\[   \|x - \widehat x_k\|_p \leq \|x - x_k\|_p + 2\epsilon k^{1/p - 1} \|x - x_k\|_1 \quad\text{and}
\]
\[ \| x - \widehat x_k\|_1 \leq (1 + 3\epsilon) \| x - x_k\|_1.
\]

\end{cor}

\subsection{The measurement matrices}
This subsection describes how to construct the measurement matrix
$\mtx{\Psi}$.  We note that this construction has the same
super-structure as the random construction in~\cite{GSTV07:HHS} and,
as such, we highlight the necessary changes while summarizing the
similar pieces.  For ease of exposition, we set $\epsilon = 1$.  The
matrix $\mtx{\Psi}$ consists of two pieces: an identification matrix
$\mtx{\Omega}$ and an estimation matrix $\mtx{\Phi}$.  We use the
first part of the sketch to identify indices of significant components
of the signal quickly and then we use the second part to estimate the
coefficients of those terms.

We use concatenated copies of the (normalized) adjacency matrices
$\mtx{\Gamma}$ of $(k',\epsilon')$-unbalanced expanders with left
degree $d$.  In what follows, we let the sparsity parameter $k'$ vary
but fix $p = 1 + 1/\log(n)$.  Normalize by the factor $d^{-1/p}$.  By
Theorem~\ref{thm:ubrip}, such matrices satisfy the
$\rip{p}{k'}{\delta}$ property for $p = 1 + 1/\log n$ and $\delta \leq
C \epsilon'$.  Fix $\epsilon'$ so that $4 \epsilon' + 1/(2d^{1/p}) <
1/2$.  Finally, we note that the number of rows in such a matrix is
$k'2^{(\log\log n)^{O(1)}}$ for known {\em explicit} constructions of
such matrices and is $k'\polylog(n)$ for {\rm random} constructions.
To simplify our accounting below, we denote the number of rows by $R$
for either the explicit or the random constructions.

\subsubsection{Identification matrix.}
The identification matrix $\mtx{\Omega}$ is a 0--1 matrix with
dimensions \linebreak $\bigO(R \polylog(n)) \times n$.  It consists of
a combination of a three structured matrices.  Formally,
$\mtx{\Omega}$ is the row tensor product $\mtx{\Omega} = \mtx{B} \rtp
\mtx{A}$.  The \term{bit-test matrix} $\mtx{B}$ has dimensions
$\bigO(\log n) \times n$, and the \emph{isolation matrix} $\mtx{A}$
has dimensions $\bigO(R\polylog n) \times n$.

The isolation matrix $\mtx{A}$ is a 0--1 matrix with a hierarchical
structure.  It consists of $\log(k)$ blocks $\mtx{A}^{(j)}$ labeled by
$j = 1, 2, 4, 8, \dots, J$, where $J = k$.  Each block, in turn, has
further substructure as a row tensor product of a collection of 0--1
matrices: $\mtx{A}^{(j)}_{r,s} = \mtx{R}^{(j)}_r \rtp \mtx{S}^{(j)}_s$
for $s = 1, 2, 4,\ldots, k$ and $r = 2s, 4s, 8s, \ldots, n$. The first
matrix $\mtx{S}^{(j)}_s$ is called the \term{sifting matrix} and the
second matrix $\mtx{R}^{(j)}_r$ is called the \term{noise reduction
  matrix}.  Each sifting matrix $\mtx{S}^{(j)}_s$ is the normalized
adjacency matrix of an $(s,\epsilon')$-unbalanced expander (or,
equivalently, a normalized 0--1 $\rip{p}{s}{\delta}$ matrix).

The purpose of the noise reduction matrix $\mtx{R}^{(j)}_r$ is to
attenuate the noise in a signal that has a single large component.  It
is also a 0--1 valued matrix constructed as follows. For each pair of
indices $(r,s)$, let $\beta\approx r/s$ be a prime power corresponding
to a finite field of size $\beta$.  Then we set each noise reduction
matrix to be the Nisan-Wigderson generator over the field of size
$\beta$.  That is, we index the rows of the matrix by by pairs $(a,b)$
of field elements and index columns by polynomials $q$ of degree at
most $\polylog(n)$.  Position $((a,b),q)$ is 1, if $q(a)=b$, and 0,
otherwise. Such a construction produces a matrix with $O(\beta^2)$
rows.  See~\cite{DeV} or~\cite{CM06:Combinatorial-Algorithms} for
similar constructions of similar sizes.

We note that the dimensions of the product of matrices
$\mtx{A}^{(j)}_{r,s} = \mtx{R}^{(j)}_r \rtp \mtx{S}^{(j)}_s$ are the
critical ones (not the dimensions of the individual matrices).  Each
block is of dimension $\bigO(R \polylog(n)) \times n$.

\subsubsection{The estimation matrix}
The estimation matrix $\mtx{\Phi}$ is an $\rip{p}{K}{\delta}$ matrix
(or, equivalently, the adjacency matrix of an
$(K,\epsilon)$-unbalanced expander).  The parameter $K = O(k
\polylog(n))$ specifies how many spikes are identified during the
identification stage.  To ensure that the columns of $\mtx{\Phi}$ have
unit $\ell_p$ norm, we scale the matrix appropriately.

\subsection{The HHS$(p)$ algorithm}
The structure of the HHS$(p)$ algorithm remains unchanged.  We include
pseudo-code in Figure~\ref{fig:algo} for completeness.  This algorithm
is an iterative procedure with several steps per iteration, including
identifying a small set of signal positions that contain a substantial
fraction of the $\ell_p$ weight, then estimating their values, and
finally subtracting their encoded contribution from the initial
measurements.
\begin{center}
\begin{figure}[tb]
\tt
\begin{tabular}{|p{.9\textwidth}|}
\hline
%\raggedright
\centerline{\textbf{Algorithm: HHS$(p)$ Pursuit}} \\
Inputs: The number $k$ of spikes, the HHS measurement matrix $\mtx{\Psi}$, \\
\hspace{4pc} the initial sketch $v = \mtx{\Psi}x$, and $\Delta$ the
dynamic range of $x$ \\
Output: A list $L$ of $\bigO(k)$ spikes \\
\\
For each iteration $i = 0, 1, \ldots, \bigO(\log \Delta)$ $\{$ \\
\hspace{2pc} For each scale $j = 1,2,4,\ldots, \bigO(k)$ $\{$ \\
\hspace{4pc} Initialize $L' = \emptyset$. \\
\hspace{4pc} For each row of $\mtx{A}^{(j)}$ $\{$ \\
\hspace{6pc} Use the $\bigO(\log n)$ bit tests to identify one spike location \\
%\hspace{7pc} in each row $l$ from measurements $\vct{r}^{(j)}_{\rm id}(1+(l-1)\log d: l \log d)$. \\
\hspace{4pc} $\}$ \\
\hspace{4pc} Retain a list $L'_j$ of the spike locations \\
\hspace{6pc} that appear $\Omega( dj )$ times each \\
\hspace{4pc} Update $L' \longleftarrow L' \cup L'_j$ \\
\hspace{2pc} $\}$ \\
\hspace{2pc} Estimate coefficients for the indices in $L'$ by forming $\Fee_{L'}^\psinv s_{\rm est}$ \\
\hspace{4pc} with Jacobi iteration \\
\hspace{2pc} Update $L$ by adding the spikes in $L'$ \\
\hspace{4pc} If an index is duplicated, add the two values together \\
\hspace{2pc} Prune $L$ to retain the $\bigO(k)$ largest spikes \\
\hspace{2pc} Encode these spikes with measurement matrix $\mtx{\Psi}$ \\
\hspace{2pc} Subtract encoded spikes from original sketch $v$ to form \\
\hspace{4pc} a new residual sketch $s$ \\
$\}$ \\
\hline
\end{tabular}
%\vspace{1pc}
\caption{Pseudocode for the HHS$(p)$ Pursuit algorithm}
\label{fig:algo}
\end{figure}
\end{center}

\subsection{Proof sketches}
%This section describes, at the highest level, why HHS$(p)$ works.  We
%establish the following result.
%\begin{thm}
%  Fix $k$.  Assume that $\mtx{\Psi}$ is a measurement matrix that
%  satisfies the conclusions of Lemmas \ref{lem:isolation},
%  \ref{lem:noisereduce}, and \ref{lem:kp-bound}.  Suppose that $x$ is
%  a $n$-dimensional signal.  Given the sketch $v = \mtx{\Psi} x$, the
%  HHS$(p)$ Pursuit algorithm produces a signal $\widehat{x}$ with at
%  most $O(k)$ nonzero entries.  This signal estimate satisfies
%$$
%\smnorm{p}{ x - \widehat{x} }
%	\leq C k^{1/p - 1} \pnorm{1}{ x - \vct{x}_k }.
%$$
%\end{thm}
%Let $k$ and $\eps$ be fixed.  Observe that we can apply the above
%theorem with $k' = k / \eps^{p/(1-p)}$ to obtain a signal estimate
%$\widehat{\vct{x}}$ with $O(k)$ terms that satisfies the error bound
%$$
%\smnorm{p}{ \vct{x} - \widehat{\vct{x}} }
%	\leq C \epsilon k^{1/p - 1} \pnorm{1}{ \vct{x} - \vct{x}_{k'} }.
%$$
%The running time increases by a factor of
%$(1/\eps^{2p/(1-p)})\polylog(1/\epsilon)$.  This leads to Theorem
%\ref{thm:hhsp}.  We give an overview of the proof in the next
%subsections. 

The goal of the algorithm is to identify a small set of signal
components that carry most of the $p$-energy in the signal and to estimate
the coefficients of those components well.  We argue that, when our
signal estimate is poor, the algorithm makes substantial progress
toward this goal.  We focus on the analysis of the
algorithm in the case when our signal estimate is poor as this is the
critical case.  More precisely, assume
that the current approximation $\vct{a}$ satisfies
\begin{equation} \label{eqn:recovery-cond}
\pnorm{p}{ \vct{x} - \vct{a} } >
	k^{1/p-1}\pnorm{1}{ \vct{x} - \vct{x}_k }.
\end{equation}
Then one iteration produces a new approximation $\vct{a}^{\rm new}$
for which $\pnorm{p}{ \vct{x} - \vct{a}^{\rm new} } \leq \frac{1}{2}
\pnorm{p}{ \vct{x} - \vct{a} }$.  In this fashion, the algorithm
improves the $\ell_p$-norm of the error geometrically until it is
small enough.  Contrast this estimate with that of the previous
section (where we reduced the $\ell_0$ norm of the error at each
step).  We will show this major result in a series of steps:
\begin{enumerate}
\item First, we obtain a fundamental relationship between the large
  and the small entries in a signal in
  Lemmas~\ref{lemma:head}, \ref{lemma:headandtail}.
\item Next, we show in Lemma~\ref{lem:isolation} that the sifting matrix
  isolates significant coefficients in a moderate amount of noise.
\item Lemma~\ref{lem:noisereduce} proves that the noise reduction matrix reduces
  the $\ell_1$-norm of the noise enough so that we can accurately
  identify the indices of the significant coefficients.
\item Finally, we use a small RIP$_{p,L,\delta}$ matrix to estimate the
  values of the coefficients in Lemma~\ref{lem:kp-bound}.
\end{enumerate}
%This way, the
%algorithm improves $\pnorm{p}{x-\widetilde x}$ geometrically until it is
%good enough and, subsequently, never causes $\pnorm{p}{x-\widetilde x}$
%to be worse than $3/2$ times worse than optimal.  (The factor $3/2$
%can be absorbed into $\epsilon$.)

We describe some generic properties of signals that are important in
the analysis.  Given a signal $g$, we write $g_t$ to denote the signal
obtained by zeroing out all the components of $g$ except the $t$
largest in magnitude (break ties lexicographically).  We refer to
$g_t$ as the {\em head} of the signal; it is the best approximation of
the signal using at most $t$ terms with respect to any monotonic norm
(such as $\ell_p$).  The vector $g - g_t$ is called the {\em tail} of
the signal.  First, we show that the $\ell_p$ norm of the tail of a
signal is much smaller than the $\ell_1$ norm of the entire signal.
\begin{lemma}
For any signal $\vct{g}$, it holds that
\[  \| \vct{g} - \vct{g}_t \|_p \leq \Big(\frac{1}{p-1}\Big)^{1/p} t^{1/p - 1} \| \vct{g}\|_1.
\]
\label{lemma:head}
\end{lemma}
\begin{proof}
  Assume that the terms $g(i)$ of $g$ are arranged in decreasing order
  $|g(1)|\ge|g(2)|\ge\cdots\ge|g(n)|$.  We observe that each term
  $\vct{g}(i)$ in $\vct{g}$
  satisfies $|\vct{g}(i)| \leq i^{-1} \|\vct{g}\|_1$ and we
  can bound the $\ell_p$ norm of the tail as
\begin{align*}
    \|\vct{g} - \vct{g}_{t}\|_p 
    &\leq \left( \sum_{i=t+1}^n \frac{1}i^{-p} \|\vct{g}\|_1^p \right)^{1/p} \\
    &\leq \left( \|\vct{g}\|_1^p \frac{1}{p-1} t^{1-p}\right)^{1/p} \\
    &\leq \left(\frac{1}{p-1}\right)^{1/p} t^{1/p-1} \|\vct{g}\|_1 
\end{align*}
\end{proof}
We note that for $p = 1 + \iota$, the constant $C_p =
\Big(\frac{1}{p-1}\Big)^{1/p}$ is bounded below by $1/\iota$ and from
above by $1/\iota^2$ so if $p = 1 + 1/\log(n)$, then $\log(n) \leq C_p
\leq \log^2(n)$.

Let $\vct{r} = \vct{x} - \vct{a}$ denote the residual signal at the
current approximation $\vct{a}$ and let $\ell = O(k)$ be the number of
terms in the approximation $a$ at each iteration.  When condition
\eqref{eqn:recovery-cond} holds, most of the $p$-energy in the signal
is concentrated in its largest components.  The next lemma relates the
$\ell_p$ norm of the residual to the $\ell_p$ norm of its tail $r -
r_\ell$ and to the $\ell_1$ norm of the entire residual.
\begin{lemma}[Head and Tail]
  Suppose that \eqref{eqn:recovery-cond} holds.  Let $\ell = O(k)$ be
  the number of terms in the approximation $\vct{a}$ at each
  iteration.  Then the following bounds hold.
\begin{eqnarray}
    \| r \|_p & \geq C_p \| r - r_\ell \|_p \\
    \| r \|_p & \geq C \ell^{1/p - 1} \| r \|_1.
\end{eqnarray}
\label{lemma:headandtail}
\end{lemma}
\begin{proof}
  Observe that because the approximation $\vct{a}$ contains no more
  than $\ell$ terms in each iteration, the number of terms in the head
  of the residual $r_\ell = \vct{a} - \vct{x}_k$ is not more than
  $\ell + k = O(k)$.  Rather than keep careful track of the exact
  number of terms in each approximation and each residual, we will, by
  abuse of notation, suppress all constants and refer to the number of
  terms as $O(k)=\ell$.  Therefore, \eqref{eqn:recovery-cond} implies
  that
\begin{equation}
\| r \|_p = \| x - a \|_p > k^{1/p - 1} \| x - x_k \|_1 = k^{1/p-1} \| (x - a) + (a - x_k)\|_1 
       = \ell^{1/p-1} \| r - r_{\ell}\|_1.
\label{eqn:tailest}       
\end{equation}
Now, apply Lemma~\ref{lemma:head} to the residual tail $r - r_{\ell}$ with $t = \ell$ to obtain
\[  \|r - r_{\ell}\|_1 \geq  C_p \, \ell^{1-1/p} \| r - r_{\ell}\|_p.
\]
If we combine this relation with our previous calculation, we find
\[  \|r \|_p \geq C_p \|r - r_\ell\|_p
\]
which is our first desired inequality.

To prove the second inequality, we use the Cauchy-Schwarz inequality to see that
\[  \| r \|_p \geq \| r_\ell \|_p \geq \ell^{1/p - 1} \| r_\ell \|_1.
\]
Finally, we add this inequality to \eqref{eqn:tailest} to obtain
\[  \| r \|_p \geq C \ell^{1/p-1} \| r \|_1.
\]
\end{proof}

Now, we turn to the identification of significant signal entries.  We
pull these off in bands of decreasing magnitude.  For exposition, we
assume that $\|\vct{g}\|_1 = 1$.  We can think of the action of one
submatrix $\mtx{S}^{(j)}_t$ as
$$
\mtx{S}^{(j)}_t : \vct{g} \mapsto
\begin{bmatrix} \vct{h}^{1} & \vct{h}^{2} & \dots & \vct{h}^{N}
\end{bmatrix},
$$
mapping each input signal to a collection of output signals.
\begin{lemma}[Isolation]
  Let $\vct{g}$ be a signal and let $I$ be a subset of $\{1, 2,
  \ldots, n\}$ with $s \leq |I| < 2s$.  Let $\ell = |I|$.  We apply
  the sifting operator $\mtx{S}^{(j)}_s$ to $\vct{g}$.  For at least
  $(1-\rho') \ell=(1 - (4 \epsilon' + 1/(2d^{1/p}))$ of the indices $i
  \in I$, there is an output signal $\vct{h}$ of the form
\[   \vct{h} = \vct{g}_i \vct{e}_i + \nu \qquad\text{and}\qquad \|\nu\|_1 \leq \frac{2}{\ell} \|\vct{g}\|_1.
\]
\label{lem:isolation}
\end{lemma}
\begin{proof}
To prove this result, we must show a sequence of shorter results.  Our goal is to isolate significant spikes from one another and to reduce the contribution of the rest of the signal to these isolations.
\begin{defn}
  An $(s,r)$ band is a band of spikes of magnitude between $1/r$ and
  $1/(2r)$.  There is some $s$ (a power of 2) such that the number of
  spikes in the band is between $s$ and $2s$.  If the
  $p$-contribution, $sr^{-p}$, of this band to the current residual
  error $\| \vct{x} - \vct{a}\|_p^p$ is greater than the average
  band's contribution, $\frac{1}{\log n} \Bigg( k^{1/p-1} \| \vct{x} -
  \vct{a}\|_p \Bigg)^p$; i.e.,
  \[ s r^{-p} \geq \frac{1}{\log n} \Bigg( k^{1/p-1} \| \vct{x} -
  \vct{a}\|_p \Bigg)^p
  \]
  then we call this band {\em significant}.
\end{defn}
The next lemma tells us how many spikes lie above a significant band.  
\begin{lemma}
  If an $(s,r)$ band is significant, then there are at most $s
  \polylog(n)$ terms of magnitude greater than $1/r$ in the current
  residual.
\end{lemma}
\begin{proof}
  If the $s$ spikes of magnitude $1/r$ have $p$-energy at least as big
  as $s'$ spikes of magnitude $1/r' > 1/r$, then
\[   s(1/r)^p \geq s' (1/r')^p > s' (1/r)^p
\]
and so the number of larger spikes must be less than $s$, $s' \leq s$,
as $p > 1$.  Because there are only $\log n$ possible $s'$, if we take
a union over all $s'$, we achieve the desired result.
\end{proof}

\begin{defn}
  We say that a spike with index $i$ is {\em isolated} from other
  spikes with indices in a set $I$ by a matrix $\mtx{\Gamma}$ if there
  is a row in $\mtx{\Gamma}$ that has a one in column $i$ and no other
  one in any column in $I$.
\end{defn}

We are ready to begin the proof of Lemma~\ref{lem:isolation}.  Our
goal is to show that $\mtx{\Gamma} = \mtx{S}^{(j)}_s$ isolates from
each other the vast majority of spikes of magnitude $\geq 1/r$, and,
hence, isolates the majority of those with magnitude equal to $1/r$
from the larger spikes.

By the expansion properties of $\mtx{\Gamma}$, at least $(1 -
4\epsilon') \ell$ distinct spikes in the original signal are isolated
$d/2$ times in the measurements.
% First, let us assume that $\vct{g}$ is exactly $\ell$-sparse.  By
% the expansion properties of $\mtx{\Gamma}$, at most $\epsilon' d
% \ell$ rows of $\mtx{\Gamma} \vct{g}$ fail to isolate a spike.  That
% is, $\mtx{\Gamma}$ has exactly $(1-\epsilon') d \ell$ isolating rows
% for $\vct{g}$.  ACG: this is the same argument as used in the
% previous section.
Next, we show that $\mtx{\Gamma}$ not only isolates spikes but also
reduces the $\ell_1$ norm of any noise, if it is present.  We will
argue that few of the good output signals have a lot of noise.  Let
$u_m$ denote the magnitude of the $m$th output position.  We observe
that the $1 \to 1$ operator norm of the matrix $\mtx{\Gamma}$ is
$d^{1-1/p}$.  By Markov's inequality,
\begin{align*}
\#    \Big\{ m \,\Big|\, u_m \geq \frac{2}{\ell} \| \vct{g} \|_1 \Big\}  &\leq  
                 \frac{\ell}{2 \|\vct{g}\|_1} \sum u_m \\
            &= \frac{\ell}{2 \|\vct{g}\|_1} d^{1-1/p} \|\vct{g}\|_1 \\
            &= \frac{d^{1 - 1/p} \ell}{2}.
\end{align*}
Therefore, there are at least $(1-\rho) d \ell/2$ isolating rows and
no more than $\frac{d \ell}{2} d^{-1/p}$ of the rows are corrupted by
noise with large $\ell_1$ norm.  This leaves $d \ell/2 ((1-4 \epsilon') -
\frac{1}{2 d^{1/p}}) = (1-\rho') d \ell/2$ good rows with a single
distinct spike and a small amount of noise.  With judicious choice of
unbalanced expander parameters, we have $(1 - \rho')\ell$ isolated spikes
corrupted by a small amount of noise.
\end{proof}

The isolation lemma~\ref{lem:isolation} produces signals which consist
of single isolated spikes plus some noise.  In order to apply
bit-testing accurately, we must reduce the $\ell_1$ norm of the noise
from a factor $1/s$ of the original $\ell_1$ norm by a further factor
$s/r$ (i.e., to a factor $1/r$ of the original noise).  Assume, for
exposition, that the original $\ell_1$ norm is one.  To achieve this
reduction, we use a noise reduction matrix $\mtx{R}^{(j)}_r$.
\begin{lemma}[Noise reduction]
  Let $\vct{h} = \vct{g}_i \vct{e}_i + \nu$ with $\|\nu\|_1 \leq 1/s$.
  We apply the noise reduction operator $\mtx{R} = \mtx{R}^{(j)}_r$ to
  $\vct{h}$.  In at least half of the output signals, we have signals
  of the form
\[   \vct{\delta} + \vct{\nu} \qquad\text{with}\qquad \|\nu\|_1 \leq C/r.
\]
\label{lem:noisereduce}
\end{lemma}
\begin{proof}
  Consider the submatrix $\mtx{R'}$ of $\mtx{R}$ obtained by
  extracting the $\beta$ rows of $\mtx{R}$ that have a 1 in column $i$ and
  then removing column $i$ itself.  By construction, $\mtx{R}'$ has at
  most $\polylog(n)$ ones in any column, so its 1-to-1 operator norm
  is at most $\polylog(n)$.  Thus the average over rows of $\mtx{R}'$
  of $\pnorm{1}{\mtx{R}'\nu}$ is at most $\frac{\polylog(n)}{\beta
    s}\approx 1/r\polylog(n)$, or at most $1/2r$ if we incorporate extra
  log factors into $\beta$.

  In addition, the number of rows in the matrix $\mtx{A}^{(j)}_{r,s}$
  is $s\cdot\max(1,(r/s)^2)$, where $s\le k$ and $sr^{-p}\ge k^{p-1}$.
  If $s/r \ge 1$, the number of rows is $s\le k$.  If $s/r \le 1$, the
  number of rows is $r^2/s$.  Observe that $sr^{-p}\ge k^{1-p}$ and
  $r/s\ge 1$ imply $1/r\ge 1/k$.  Since $sr^{-p}\ge k^{1-p}$, it
  follows that
  \[ \frac{r^2}{s} = \frac{r^p}{s r^{p-2}} \leq
  \frac{k^{p-1}}{r^{p-2}} \leq k^{2-p} k^{p-1} = k.
\]
\end{proof}

The estimation portion of the algorithm is exactly the same as the HHS
algorithm.  Let $L$ be the set of identified candidates, $|L|\le K$.
By the RIP-$p$ property, the submatrix of $\Phi$ given by columns indexed
by $L$ is invertible and the inverse $\Phi_L^{+}$ is bounded in the
appropriate operator norm.  we can apply $\Phi_L^{+}$ to $\Phi x$ by
Jacobi iteration in time $O(K^2 \polylog(n))$ to estimate the
coefficients in $L$.  In the exact case of $x_{\overline L}=0$, we
recover $x$ exactly.  Otherwise, the RIP-$p$ property assures that the
$p$-norm of the vector of estimates is bounded in terms of the
$p$-norm of $x-x_{K}$; i.e., this estimation process introduces errors
that are small compared with the error associated with missing the
$n-K$ other terms in $x$.  Its proof is a small modification of the
proof in~\cite{GSTV07:HHS}, which itself was originally proven by
Rudelson.
\begin{lemma}
  Suppose that $\Ph$ is an $m \times n$ RIP$(p,K,\delta)$ matrix.
  Then the following holds for every vector $x \in \R^n$.
  \[ \| \Ph x\|_p \leq (1 + \delta) \Big( \|x\|_p + K^{1/p - 1} \| x
  \|_1 \Big).
  \]
\label{lem:kp-bound}
\end{lemma}

For completeness, we conclude this section by showing that a mixed
norm error bound of the form $\ell_p \leq C/k^{1-1/p} \ell_1$ gives us
an $\ell_1$ error bound $\ell_1 \leq C \ell_1$.
\begin{thm}
Suppose
\[\pnorm{p}{\widetilde x-x}\le \epsilon k^{1/p-1}\pnorm{1}{x_k-x}.\]
Then
$\pnorm{p}{\widetilde x_k-x}\le \pnorm{p}{x_k-x}
   +2\epsilon k^{1/p-1}\pnorm{1}{x_k-x}$ and
$\pnorm{1}{\widetilde x_k-x}\le (1+3\epsilon)\pnorm{1}{x_k-x}$.
\end{thm}
\begin{proof}
Let $H$ denote $\supp{x_k}\cup\supp{\widetilde x_k}$, so that
$|H|\le 2k$.  Let $H^{{\rm c}}$ denote the complement of $H$.
\begin{eqnarray*}
\pnorm{p}{\widetilde x_k-x}
& \le & \pnorm{p}{(\widetilde x_k-x)|_H}
       +\pnorm{p}{(\widetilde x_k-x)|_{H^{{\rm c}}}}\\
& \le & \pnorm{p}{(\widetilde x_k-\widetilde x)|_H} + \pnorm{p}{(\widetilde x-x)|_H}
       +\pnorm{p}{(\widetilde x_k-x)|_{H^{{\rm c}}}}\\
& \le & \pnorm{p}{(x_k-\widetilde x)|_H} + \pnorm{p}{(\widetilde x-x)|_H}
       +\pnorm{p}{(\widetilde x_k-x)|_{H^{{\rm c}}}}\\
& \le & \pnorm{p}{x_k-x} + 2\pnorm{p}{(\widetilde x-x)|_H}\\
& \le & \pnorm{p}{x_k-x} + 2\pnorm{p}{\widetilde x-x}\\
& \le & \pnorm{p}{x_k-x} + 2\epsilon k^{1/p-1}\pnorm{1}{x_k-x}.
\end{eqnarray*}

Similarly,

\begin{eqnarray*}
\pnorm{1}{\widetilde x_k-x}
&  =  & \pnorm{1}{(\widetilde x_k-x)|_H}
       +\pnorm{1}{(\widetilde x_k-x)|_{H^{{\rm c}}}}\\
& \le & \pnorm{1}{(\widetilde x_k-\widetilde x)|_H} + \pnorm{1}{(\widetilde x-x)|_H}
       +\pnorm{1}{(\widetilde x_k-x)|_{H^{{\rm c}}}}\\
& \le & \pnorm{1}{(x_k-\widetilde x)|_H} + \pnorm{1}{(\widetilde x-x)|_H}
       +\pnorm{1}{(\widetilde x_k-x)|_{H^{{\rm c}}}}\\
& \le & \pnorm{1}{x_k-x} + 2\pnorm{1}{(\widetilde x-x)|_H}\\
& \le & \pnorm{1}{x_k-x} + 2(2k)^{1-1/p}\pnorm{p}{(\widetilde x-x)|_H}\\
& \le & \pnorm{1}{x_k-x} + 2(2k)^{1-1/p}\pnorm{p}{\widetilde x-x}\\
& \le & \pnorm{1}{x_k-x} + 2^{2-1/p}\epsilon\pnorm{1}{x-x_k}\\
& \le & (1+3\epsilon)\pnorm{1}{x_k-x}.
\end{eqnarray*}

\end{proof}

\end{document}